\documentclass[%draft, %for not including graphics %space hidden by book binding
captions=nooneline, %do not distinguish between one or more lines in captions
bibliography=totoc, %add bibliography to table of contents
numbers=noenddot, %removes points for special parts (e.g. appendix)
parskip=half, %leaves at least 1em of a row free before a new paragraph, also sets parindent to 0
headings=normal, %for smaller headings
abstracton,
12pt %for abstract
]{scrartcl} %automatically a4paper, 11pt, oneside, titlepage

\usepackage[T1]{fontenc} %for correct hyphenation and T1 encoding
\usepackage[USenglish]{babel} %for US english line breaks
\usepackage{graphicx} %for including figures
\usepackage{makeidx} % for index (Stichwortverzeichnis)
\usepackage{remreset} %for resetting numberings, e.g. the footnote number
\usepackage[nouppercase]{scrpage2} %package for page style with not only upper letters in the head
\usepackage{amsmath} %sophisticated mathematical formulas with amstex (includes \text{})
\usepackage{amssymb} %sophisticated mathematical symbols with amstex (includes \mathbb{})
\usepackage[
thmmarks, %for automatical placement of endmarks
amsmath, %for compatibility with amsmath
hyperref %for compatibility with hyperref
]{ntheorem} %for sophisticated theorem environments and proofs
\usepackage{bm} %for bold math symbols
\usepackage{bbm} %only for indicator functions
\usepackage{enumitem} %for automatic numbering of new enumerate environments
\usepackage[hang]{subfigure} %for subfigures
\usepackage{wrapfig} %for wrapping text around figures/tables
\usepackage{tabularx} %for special table environment (tabularx-table)
\usepackage{multirow}
\usepackage{dcolumn} %for aligning tabular entries with respect to a certain char, e.g. decimal point
\usepackage{booktabs} %for table layout
\usepackage{listings} %for including programming code
\usepackage{psfrag} %for editing text in (e)ps files
\usepackage{multicol}
\usepackage{longtable}
\usepackage{ltxtable}
\usepackage{supertabular}
\usepackage[table]{xcolor}
\usepackage{lscape}
\usepackage{caption}
\usepackage{rotating}
\usepackage{color}
\usepackage{chicago}
\usepackage{caption}
\usepackage{textgreek}

\usepackage{appendix}

\usepackage[
setpagesize=false, %necessary in order to not change text-/paperformat for the document
pdfborder={0 0 0}, %removes border around links
pdfpagemode=UseOutlines, %top level bookmarks are visible
]{hyperref} %all links stay black and are thus invisible
\newcommand\mytoday{\number\year-\ifcase\month\or 01\or 02\or 03\or 04\or 05\or 06\or 07\or 08\or 09\or 10\or 11\or 12\fi-\ifcase\day\or 01\or 02\or 03\or 04\or 05\or 06\or 07\or 08\or 09\or 10\or 11\or 12\or 13\or 14\or 15\or 16\or 17\or 18\or 19\or 20\or 21\or 22\or 23\or 24\or 25\or 26\or 27\or 28\or 29\or 30\or 31\fi} %see LaTeX book pg. 348
\makeatother
\setkomafont{sectioning}{\normalcolor\bfseries} %set headings as in standard class
\clearscrheadfoot %removes old settings for head and foot
\automark[section]{subsection} %automatically creates head text based on the subsections and, if no subsection there, sections
\setkomafont{captionlabel}{\bfseries} %uses \bfseries for the caption label
\newcolumntype{d}[2]{D{.}{.}{#1.#2}} %aligns the entry at "." of a column if "d" is used as column type where the arguments specify the number of digits to the left and right for which space is kept in a column (always set to the maximal numbers that appear in the whole table)
\makeatletter
\renewcommand{\p@enumii}[1]{\theenumi(#1)}
\makeatother

\theoremstyle{break} %adds a newline after the heading of a theorem environment
\theoremheaderfont{\bfseries}
\theorembodyfont{\itshape}
\theoremseparator{}
\newtheorem{definition}{Definition}[section] %number all environments in a sequence (for every section) (order does not play a role)

\newtheorem{theorem}[definition]{Theorem}

\newtheorem{remark}[definition]{Remark}
\newtheorem{example}[definition]{Example}
\theoremstyle{nonumberbreak} %adds a newline after the heading of a theorem environmentt which is not numbered
\theoremsymbol{$\Box$}
\newtheorem{proof}{Proof}

\newcommand*{\IR}{\mathbb{R}}

\newcommand{\bmu}{\mbox{\boldmath $\mu$}}
\newcommand{\btheta}{\mbox{\boldmath $\theta$}}

\def\vph{\varphi}
\def\p{\partial}
\def\oo{\infty}
\def\eps{\epsilon}
\def\s{\sigma}
\def\S{\Sigma}
\def\E{{\rm E}\,}
\def\x{{\bf x}}
\def\half{{\textstyle {1\over 2}}}
\def\rt#1{\sqrt{#1}\,}
\makeindex

\begin{document}

%\leading{18pt}
\title{Simplified Pair Copula Constructions --- Limits and Extensions}
\author{Jakob St\"ober \thanks{Center for Mathematical
Sciences, Technische Universit\"at M\"unchen, Germany. Corresponding author email: \texttt{stoeber@ma.tum.de}}  \and Harry Joe
\thanks{Department of Statistics, University of British Columbia, Canada.} \and Claudia
Czado\footnotemark[1] }
\maketitle		
\begin{abstract}
{\footnotesize So called pair copula constructions (PCCs), specifying multivariate
distributions only in terms of bivariate building blocks (pair copulas),
constitute a flexible class of dependence models. To keep them tractable for
inference and model selection, the simplifying assumption that copulas of
conditional distributions do not depend on the values of the variables which
they are conditioned on is popular.  In this paper, we show for which
classes of distributions such a simplification is applicable, significantly
extending the discussion of \shortciteN{haff2010b}. In particular, we show
that the only Archimedean copula in dimension $d\geq4$ which is of the
simplified type is 
%the multivariate MTCJ copula 
that based on the gamma Laplace
transform or its extension, while the Student-t copula is the only one arising from a scale mixture of Normals.
Further, we illustrate how PCCs can be adapted for situations where conditional copulas depend on values which are conditioned on.}
\end{abstract}	

	
\section{Introduction}

Growing capabilities to store large data sets and increasing computing power for their computational analysis substantially increased the demand to develop statistical methodology for many dependent variables over the last years.
The properties of high-dimensional distributions in general, however, remain hard to understand intuitively. 
Given that, the early focus was on the multivariate Normal distribution, which allows to describe the whole dependence structure by specifying bivariate correlations. These linear dependencies between just two variables remain easy to understand and communicate. 
But more recent developments, in particular in econometrics (c.f.~\citeN{longin2001}, \citeN{ang2002b}), showed that many data sets in higher dimensions exhibit more complex dependence structures. 

To deal with these challenging structures and keeping the model understandable at the same time, a popular technique is to sequentially decompose the joint distributions into bivariate building blocks by conditioning. These are then modeled using bivariate copulas of which many parametric families are well studied, see for example the monographs by \citeN{Joe} and \citeN{Nelsen}. Distributions arising from such a decomposition are called pair copula constructions (PCCs).
The question for which classes of multivariate models and under which assumptions such a sequential conditioning is possible, however, remained unsolved. 
%Also whether they can be extended to include or at least approximate distributions where it is not is still unclear. 
The only publication in this direction is \shortciteN{haff2010b}, providing first illustrative examples.

In this paper, we fill in this gap by classifying multivariate distributions in terms of the dependence properties of their conditional distributions (see Figure \ref{overview_graphic}). Special attention will be paid to the \textbf{simplifying assumption} that the copulas corresponding to conditional distributions are constant irrespective of the values of variables that  they are conditioned on. This assumption is usually made to keep model selection and inference tractable, and we will show for which multivariate distributions it holds.
%We show that this assumption holds only for the multivariate MTCJ copula in the Archimedean case, and also not for all elliptical distributions. As an illustration, we demonstrate the decomposition of the multivariate Normal, Student-t and MTCJ copula into bivariate building blocks. 
Furthermore, we discuss how constructions based on simplified PCCs can be extended when the simplifying assumption is not applicable.

The remainder is structured as follows.
Section \ref{pccs} introduces PCCs. In Section \ref{archimedean} and \ref{elliptical}, we demonstrate which Archimedean and elliptical copulas are simplified PCCs.
Section \ref{non-simplified} provides the remaining examples for the classification in Figure \ref{overview_graphic} and considers the increased flexibility gained by weakening the simplifying assumption. Section \ref{conclusion} concludes with an outlook to areas of future research.
\begin{figure}[!ht]
\centering
\vspace{1cm}
\includegraphics[width=1\textwidth]{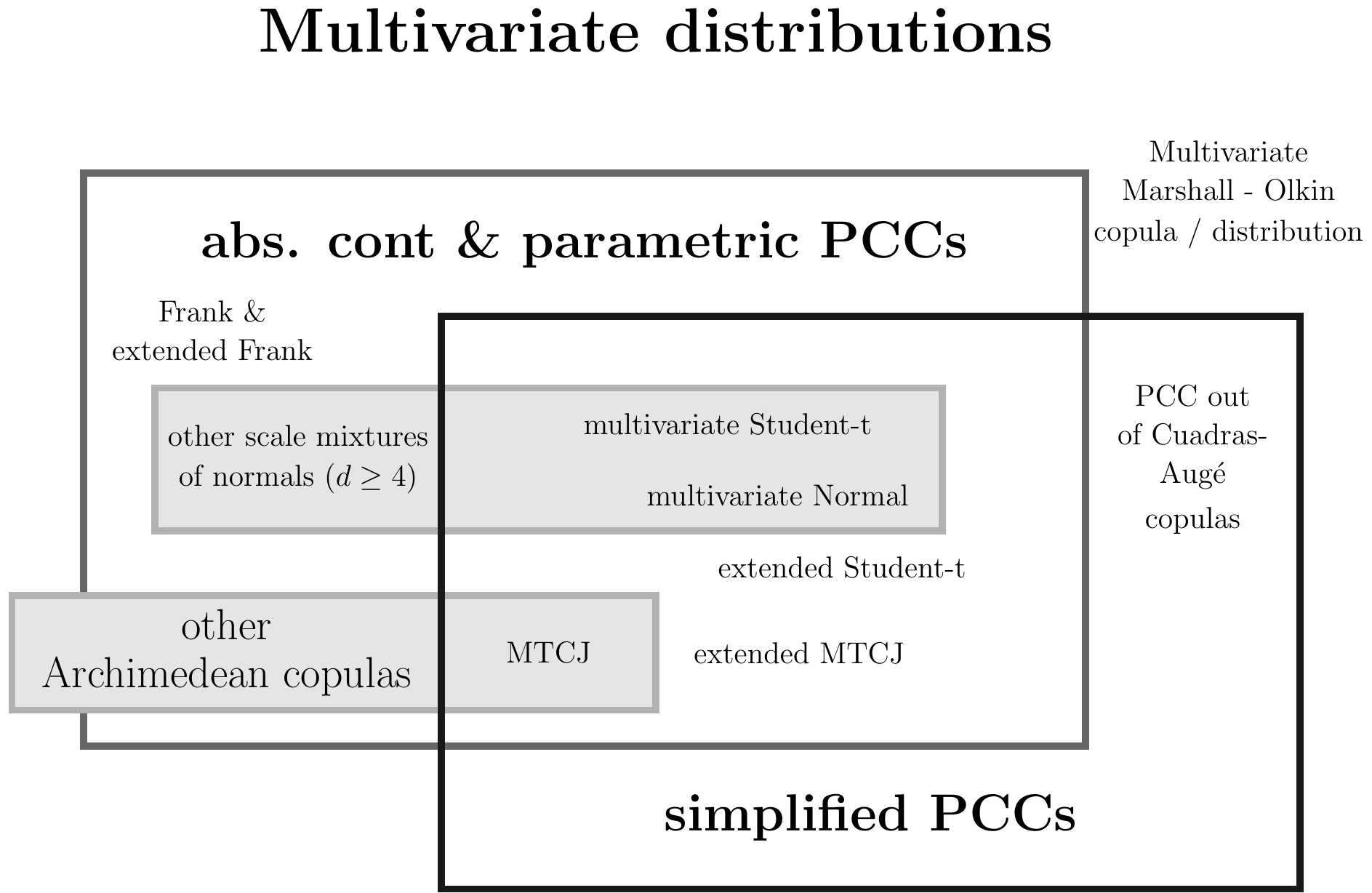}\caption{Overview of the PCC classes which are discussed including their relations and examples.}
\label{overview_graphic}
\end{figure}

\section{Pair copula constructions}\label{pccs}

A copula is defined as the cumulative distribution function (cdf) of a
multivariate distribution with uniform$(0,1)$ margins. These have become of
great importance in multidimensional dependence modeling, since the theorem
of \citeN{sklar} allows to separate the specification of a $d$-dimensional
distribution $F$ into specifying the univariate marginal distribution
functions $F_i$, $i=1,\ldots,d$, and a copula $C$, i.e., to find $C$ such that 
\begin{equation*}
F_{1:d}(x_1,\ldots,x_d)=C_{1:d}(F_1(x_1),\ldots,F_d(x_d)).
\end{equation*}
Here, we use the abbreviation $i:j:=(i,i+1,\ldots,j)$ for $i<j$.
In particular bivariate copula families are well studied, and every multivariate distributions can be constructed from marginal distributions and two-dimensional copulas by sequential conditioning, a so-called pair-copula construction.
While different principles for PCCs such as regular vines (R-vines) (\citeANP{bedford2001} \citeyear{bedford2001,bedford2002}) or non-Gaussian directed acyclic graphs (DAGs) \shortcite{bauer2011} have been investigated over the years, the underlying basic idea is always that of \citeN{joe1996}:

Let us consider three random variables (rvs) $X_1$, $X_2$ and $X_3$ with corresponding cumulative distribution functions (cdfs) $F_1, F_2$ and $F_3$, then we can first construct a joint distribution of $(X_1,X_2)$ and of $(X_3,X_2)$ by assigning copula functions $C_{12}$ and $C_{32}$, respectively:
\begin{equation}
\label{eq:construction1}
\begin{split}
F_{12}(x_1,x_2)&:=C_{12}\left(F_1(x_1),F_2(x_2)\right)\\
F_{32}(x_3,x_2)&:=C_{32}\left(F_1(x_3),F_2(x_2)\right)\\
\end{split}
\end{equation}
These bivariate distributions can be combined to a three-dimensional distribution by assigning a conditional copula $C_{13\vert 2}$,
\begin{equation}
\label{eq:construction2}
F_{123}(x_1,x_2,x_3):=\int_{-\infty}^{\infty} C_{13\vert 2}(F_{1\vert 2}(x_1\vert x_2), F_{3\vert 2}(x_3\vert x_2);x_2) \  dF_2(x_2).
\end{equation}
For every cdf $F_{123}$, Sklar's theorem implies the existence of a bivariate copula that couples $F_{1\vert 2}(\cdot \vert x_2)$ and $F_{3\vert 2}(\cdot \vert x_2)$ for every $x_2$.
The generalization of this principle to arbitrary dimensions is obvious and the conditional copulas corresponding to such a construction can be determined for every multivariate distribution (c.f.~\citeN{patton2006} for a discussion of conditional copulas).

To make PCCs tractable for inference and model selection however, further assumptions regarding the conditional copulas have to be made. \\
We will call a multivariate distribution an \textbf{absolutely continuous and parametric} PCC if all bivariate copula families occurring in the construction have densities with a parameter vector. In this case, it will be possible to express the likelihood function of the joint copula model in terms of the bivariate densities and to make inference about the finite-dimensional parameter based on this expression.
One example for this would be the decomposition of a $d$-dimensional density $f$ with marginal densities $f_k$, $1\leq k \leq d$, as a C-vine copula (see \shortciteN{aas2009}): with univariate parameters $\theta_k$ and copula parameters $\theta_{j,j+i}$, and $\btheta=(\theta_1,\ldots,\theta_d;\theta_{j,j+i},j=1,\dots,d-1, i=1,\dots,d-j)$,
\begin{equation*}
\begin{split}
 f_{1:d}(x_1, & \dots,  x_d; \btheta)=\prod_{k=1}^d f_k(x_k; \theta_k)\ \cdot \\
& \prod_{j=1}^{d-1}\prod_{i=1}^{d-j} c_{j,j+i\vert 1:(j-1)}\left( F_{j \vert 1:(j-1)}(x_j\vert \textbf{x}_{1:(j-1)}),F_{j+i \vert 1:(j-1)}(x_{j+i}\vert \textbf{x}_{1:(j-1)});\theta_{j,j+i}  \right).
\end{split}
\end{equation*}
The above multivariate density is valid for a \textbf{simplified} PCC
because all bivariate copulas occurring in the constructions do not depend
on the variables that are conditioned on, e.g.,
\begin{equation*}
C_{13\vert2}(\cdot,\cdot;x_2) = C_{13\vert2}(\cdot,\cdot).
\end{equation*}
%\begin{remark}
%Given a particular class of PCCs such as non-Gaussian DAGs or R-vines, the above definitions can be rephrased, requiring the condition only for a subset of the bivariate conditional distributions to be used in the construction. In this way the term simplified PCC has already been used for R-vine PCCs by \shortciteN{haff2010}. 
%Within the scope of this paper however all distributions of interest will be exchangeable such that both definitions are equivalent.
%\end{remark}
This \textbf{simplifying} assumption reduces the specification of a PCC to choosing bivariate copula families and their parameters. This means that the potentially complex dependence between variables that are conditioned on and the copula functions can be neglected, and model selection can be performed more easily.
Also other techniques for PCCs such as for R-vines (stepwise estimation in \citeN{haff2010} and \shortciteN{dissmann2011}, inference as in \shortciteN{aas2009} or Bayesian techniques as in \citeN{min2010}) extend to the non-simplified case only if the dependence on the variables that are conditioned on is of a simple form (we discuss this in Section \ref{non-simplified}, see in particular Example \ref{ex:frank}).

\section{Archimedean copulas}\label{archimedean}
In this section, we characterize the Archimedean copulas that  are simplified PCCs.
A $d$-dimensional Archimedean copula is given by 
\begin{equation}
C_{1:d}(u_1,\dots,u_d)=\varphi\Big(\sum_{j=1}^d \varphi^{-1}(u_j) \Big),
\end{equation} 
for an Archimedean generator function $\varphi \in \mathcal{L}_d$. Here, $\mathcal{L}_d$ is the class of functions $\varphi : [0,\infty) \mapsto [0,1]$, which are strictly decreasing on $[0,\inf\{s, \varphi(s)=0 \}]$ (with $\varphi(0)=1$, $\varphi(\infty)=0$) and differentiable on $[0,\infty)$ up to order $d-2$, such that $(-1)^j \varphi^{(j)} \geq 0, j=1,\dots,d-2$, and where further $(-1)^{d-2} \varphi^{(d-2)}$ is non-increasing and convex (\citeN{mcneil2009}, \citeN{Joe}, \citeN{mueller2005}). 
These copulas have the appealing property that the conditional distribution function is given in a simple analytical form which facilitates the study of their properties. 
%A popular example is the (generalized) , with generator function
%\begin{equation}\label{mtcj_generator}
%\varphi_{\Gamma}(s;\theta)=\left(1+\theta \cdot s \right)_+^{-1/\theta}, \ \theta \geq -\frac{1}{d-1},
%\end{equation}
%in dimension $d$, where $(y)_+$ denotes $\max(0,y)$.
Given a Gamma random variable with shape parameter $p$, rate parameter $b$, density
\begin{displaymath}
f_{\Gamma}(x;b,p) = \frac{b^p}{\Gamma(p)}x^{p-1}e^{-b x},
\end{displaymath}
% HJ reformatted fractions so that font is not so small
mean $p/b$ and variance $p/b^2$, %and %moment generating function 
%\begin{displaymath}
%M_{\Gamma}(s;b,p) =\left(1-\frac{s}{b}\right)^{-p},
%\end{displaymath}
%this generator is, for $\theta \geq 0$, the 
the Laplace transform (LT) is
\begin{displaymath}
\varphi_{\Gamma}(s;\theta)=\int_0^\infty e^{-sx} f_{\Gamma}(x;1/\theta,1/\theta) dx=(1+\theta s)^{-1/\theta}, \ \ \theta \geq 0.
\end{displaymath}
The Archimedean copula corresponding to this LT as generator function is called MTCJ\footnote{This copula is also called Clayton copula due to its appearance in \cite{clayton1978}. It is the copula of the multivariate Pareto distribution \cite{mardia1962} and of the multivariate Burr distribution \cite{takahasi1965}. It was first mentioned as a multivariate copula in \citeN{cook1981} and as a bivariate copula in \citeN{kimeldorf1975}. The extension to negative dependence was given in \citeN{Genest1986} for the bivariate case and in \citeN[pp. 157-158]{Joe} for the multivariate case. Since many properties were discovered studying the corresponding distribution function and \citeN{cook1981} mentioned its general form we call it MTCJ copula.} copula. 
From the form of the conditional distributions derived in \citeN{takahasi1965} it is obvious that the copulas corresponding to bivariate conditional margins of the MTCJ copula are again MTCJ copulas. The above gamma LT family extends to Archimedean generators (for the {\it generalized} MTCJ family)
\begin{equation}\label{mtcj_generator}
\varphi_{\Gamma}(s;\theta)=\left(1+\theta \cdot s \right)_+^{-1/\theta}, \ \ \theta \geq -\frac{1}{d-1}, \ \ (x)_+=\max\{0,x\}.
\end{equation}
The conditional generator (c.f. \citeN{mesfioui2008}) when conditioning on $m$ variables becomes
\begin{equation}\label{mtcj_cond_generator}
\varphi_{m}(s;\theta)=\left(1+\frac{s \theta}{m \theta +1}\right)_+^{- (m \theta +1) / \theta}.
\end{equation} 
For $\theta \geq 1/(d-1)$, $\theta/(m\theta +1) \geq 1/(d-1-m)$, so the parameter range is consistent, and the (generalized) MTCJ copula is a simplified PCC where all bivariate building blocks are MTCJ with corresponding choice of parameters. 
In fact, under a weak regularity assumption, it is the only multivariate Archimedean copula that constitutes a simplified PCC.

\begin{theorem}\label{theorem:archimedean}
A $d$-dimensional Archimedean copula with 
a generator that is twice continuously differentiable on the set where
it is positive
%$2$-times differentiable generator function 
is a simplified PCC if and only if its generator is in the family (\ref{mtcj_generator}).
\end{theorem}
\begin{proof}
Appendix \ref{appendix:archimedean}.
\end{proof}
Note that the differentiability condition for the generator is satisfied for all Archimedean copulas if $d \geq 4$.
This adds a further aspect making the MTCJ copula unique\footnote{The MTCJ copula also is the only Archimedean copula invariant under truncation, in the sense that for a rv $\textbf{U}\sim C$, C also is the copula of $\textbf{U}\vert \textbf{U}\leq \textbf{a}$, $\textbf{a} \in \left[0,1\right]^{\text{dim}(\textbf{U})}$ \cite{ahmadijavid2009}.} among Archimedean copulas.
Using different parameters than obtained through Equation (\ref{mtcj_cond_generator}), PCCs with MTCJ copulas with freely chosen parameters as building blocks naturally generalize the MTCJ copula to what we call an \textbf{extended MTCJ} copula.

\section{Elliptical copulas}\label{elliptical}
In this section, we characterize the elliptical copulas that have all conditional distributions in location-scale families. We show that not all elliptical copulas are simplified PCCs and characterize the scale mixtures of Normals which are of the simplified type.
By referring to a $d$-dimensional elliptical copula we mean a copula arising from an elliptical distribution such as the multivariate Normal or Student-$t$ distribution. A multivariate distribution is elliptical if its characteristic function has the form
\begin{displaymath}
\text{\textphi}_\textbf{X}(\textbf{t};\boldsymbol \mu, \Sigma)=\Psi(\textbf{t}' \Sigma \textbf{t}) e^{i \textbf{t}'{\boldsymbol \mu}},
\end{displaymath}
for $\Psi: \IR_0^+ \mapsto \IR$, $\bmu \in \IR^d$, and positive definite $\Sigma \in \IR^{d \times d}$.
If the distribution has a density, this implies that it is given by 
\begin{displaymath}
f_{\textbf{X}}(\textbf{x})= \lvert \Sigma \rvert^{-1/2} g\Big((\textbf{x}-\bmu)' \Sigma^{-1} (\textbf{x} - \bmu) \Big), 
\end{displaymath}
for a generator function $g:\IR^+_0 \mapsto \IR^+_0$, which can be uniquely determined from $\Psi$, see \shortciteN{cambanis1981}. For the two examples we mentioned, the generator functions have the form
\begin{displaymath}
g_{\text{Gauss},n}(t)=\frac{1}{(2\pi)^{n/2}} e^{-t/2},  \ \ \ g_{\text{Student}-t,n,\nu}(t)=\frac{\Gamma\left(\frac{\nu +n}{2}\right)}{\Gamma\left(\frac{\nu}{2}\right) (\nu\pi)^{n/2}} \cdot \left( 1+\frac{t}{\nu}\right)^{-\left(\nu +n\right)/2},
\end{displaymath} and lead to simplified PCCs. 

\vspace{.5cm}
\begin{theorem}\label{elliptical1}
The multivariate Gaussian distribution and the multivariate Student-$t$ distribution are PCCs of the simplified form.
\end{theorem}
For the Gaussian distribution this is quite obvious since all conditional
distributions are again Gaussian and do depend on the values of the
variables that they are conditioned on only through their mean, i.e., they all are in the same location family. For the Student-$t$ distribution the covariances also depend on these values, but only through a scaling factor such that all conditional distributions remain in the same location-scale family. Since changes of location and scale do not affect the copula of a multivariate distribution, this implies that the Student-$t$ distribution is a simplified PCC. For clarification, we derive the explicit form of the copulas corresponding to bivariate conditional distributions in Appendix \ref{appendix:student1}.
Just as for the MTCJ copula, changing the degrees of freedom in the bivariate building blocks of the PCC leads to a natural extension of the multivariate Student-t distribution in which different bivariate conditional margins can have different degrees of freedom.

While the copulas corresponding to these distributions are the most popular examples they also have a unique position within the class of elliptical distributions as the following theorems show.
\vspace{.5cm}
\begin{theorem}\label{elliptical_theorem}
Let us assume that the generator function $g(\cdot)$ of the density of a $d$-dimensional elliptical distribution is differentiable. 
Then, the conditional distributions remain within the same location-scale family for all values of the variables that are conditioned on, if and only if
\begin{itemize}
\item[(a)] the support of $g$ is $\IR$ and the distribution is the multivariate Student-t (or Pearson type VII) distribution or in its limiting case the multivariate Normal distribution or
\item[(b)] $g$ has compact support and $g(t;\zeta)=(1-t)_+^\zeta$, for $\zeta > 1$ up to rescaling (the Pearson type II distribution).
\end{itemize}
\end{theorem}
\begin{proof}
Appendix \ref{appendix:student2}.
\end{proof}
Here, case $(b)$ can be seen as an analogue to the extension of the MTCJ to negative dependence.

The proof that the Student-t distribution is a simplified PCC relied on the fact that in this case, conditioning only affects the location and scale of the distribution. Theorem \ref{elliptical_theorem} now shows that the t-distribution is the only elliptical distribution where this proof strategy is successful. Just as the MTCJ can be constructed using a Gamma mixture, also the Student-t distribution arises from the multivariate Normal distribution using a Gamma mixture for the square of the inverse scale parameter. For the distribution function of the MTCJ copula, we have
\begin{displaymath}
C_{MTCJ}(\textbf{u}_{1:d};\theta)=\int_0^\infty \prod_{i=1}^d \left[G_{MTCJ}(u_i;\theta) \right]^\alpha f_\Gamma(\alpha;1,1/\theta) d\alpha,
\end{displaymath}
where $G_{MTCJ}(\cdot;\theta) = \exp\left\{ -\left(\varphi_{MTCJ} \right\}^{-1} (\cdot;\theta ) \right)$ is a cdf on $\left[0,1\right]$, c.f. \citeN[p. 86]{Joe}, while 
\begin{displaymath}
f_{t,d}(\textbf{x}_{1:d}; R, \nu) =  (2\pi)^{-d/2} |R|^{-1/2} \int_0^\oo w^{d/2} \exp\{ -\half w\x'R^{-1}\x \} f_\Gamma(w; \nu/2, 2/\nu)\,dw,
\end{displaymath}
holds for the density of a multivariate Student-t distribution.
When only considering scale mixtures of Normals, we obtain the following theorem. 
\begin{theorem}\label{elliptical_theorem2}
Consider a $d$-dimensional scale mixture of Normals which is a simplified PCC
\begin{itemize}
\item[a)] in $d\geq4$ or
\item[b)] for all correlation matrices $\Sigma$, 
\end{itemize}
then it is a simplified PCC if and only if the mixing distribution is the Gamma distribution.
\end{theorem}
\begin{proof}
Appendix \ref{appendix:student3}.
\end{proof}
This contradicts the claim made in Example 4.1 of \citeN{haff2010} that all elliptical distributions are simplified PCCs.

Although this shows that the intersection between simplified PCCs and elliptical distributions is limited, it turns out that the deviation from the simplifying assumption can be ignored in many applications. In fact, the values of Kendall's $\tau$ corresponding to bivariate conditional distributions are independent of the variables that are conditioned. \shortciteN{cambanis1981} showed that the conditional correlation coefficient is equal to the partial correlation for elliptical distributions, and \shortciteN{lindskog2003} demonstrated that the relationship between the correlation coefficient $\rho$ and Kendall's $\tau$ for the Normal distribution
\begin{displaymath}
\tau = \frac{2}{\pi} \text{arcsin}(\rho),
\end{displaymath}
holds for all atom-free elliptical distributions.

\section{Non continuous and non-simplified PCCs}\label{non-simplified}

In order to complete the list of examples for the classification in Figure \ref{overview_graphic}, we need to find distributions which are not absolutely continuous but simplified PCCs and distributions which are neither absolutely continuous nor simplified PCCs. 
For the first case, let us consider a PCC with bivariate Cuadras-Aug\'e (which are a special case of the bivariate Marshall-Olkin (MO) copula, where the distribution is exchangeable) copulas as building blocks. Exempli gratia,
\begin{equation*}
C_{12}(u_1,u_2)=\min(u_1,u_2) \max(u_1,u_2)^{1-\alpha_{12}}, \ \ \alpha_{12} \in [0,1],
\end{equation*}
c.f.~\citeN[Chapter 1]{mai2012}. For the joint distribution of $(U_1,U_2,U_3)$, $F_{123}$, this implies that
\begin{equation*}
\begin{split}
F_{123}(u_1&,u_2,u_3)\\
&=\int_0^1 \min(F_{\alpha_{12}}(u_1\vert v_2),F_{\alpha_{32}}(u_3\vert v_2)) \max(F_{\alpha_{12}}(u_1\vert v_2),F_{\alpha_{32}}(u_3\vert v_2))^{1-\alpha_{13\vert2}} dv_2,
\end{split}
\end{equation*}
where the conditional cdf is given by \cite[Example 1.8]{mai2012}
\begin{equation*}
F_{\alpha}(u\vert v)=\begin{cases}
u (1-\alpha) v^{-\alpha} &u<v,\\
u^{1-\alpha} & u \geq v.
\end{cases}
\end{equation*}
Note that here $F_{123}$ is a copula.
A scatterplot of the bivariate $13$ margin of this distribution is shown in the left panel of Figure \ref{moplots} for illustration.

While this copula is of simplified form by construction, the multivariate MO copula (c.f.\ \cite[Chapter 3]{mai2012}), 
\begin{displaymath}
C_{MO}(\textbf{u}_{1:d})=\prod_{\emptyset \not = I \subset \{1,\ldots,d\}} \min_{k\in I} \left\{u_k^{\frac{\lambda_I}{\sum_{J:k\in J} \lambda_J}}\right\}, \ \ \ \lambda_I \geq 0, \ \ \sum_{J:k\in J} \lambda_J > 0 \ \ \forall k = 1,\ldots,d,
\end{displaymath} in contrast is neither a simplified nor an absolutely continuous PCC. We derive the copula of bivariate conditional distributions for the trivariate case in Appendix \ref{section:mo}, the cdf of the conditional copula of $U_1,U_2\vert U_3=u_3$ is plotted in the right hand panel of Figure \ref{moplots} for different values of the conditioned $U_3$. Note that, due to the non continuity of $F_{U_1\vert U_3}$ and $F_{U_2 \vert U_3}$ this copula is uniquely defined only outside the shaded areas.

\begin{minipage}[c]{\textwidth}
\begin{minipage}[c]{0.5\textwidth}
\includegraphics[width=\textwidth]{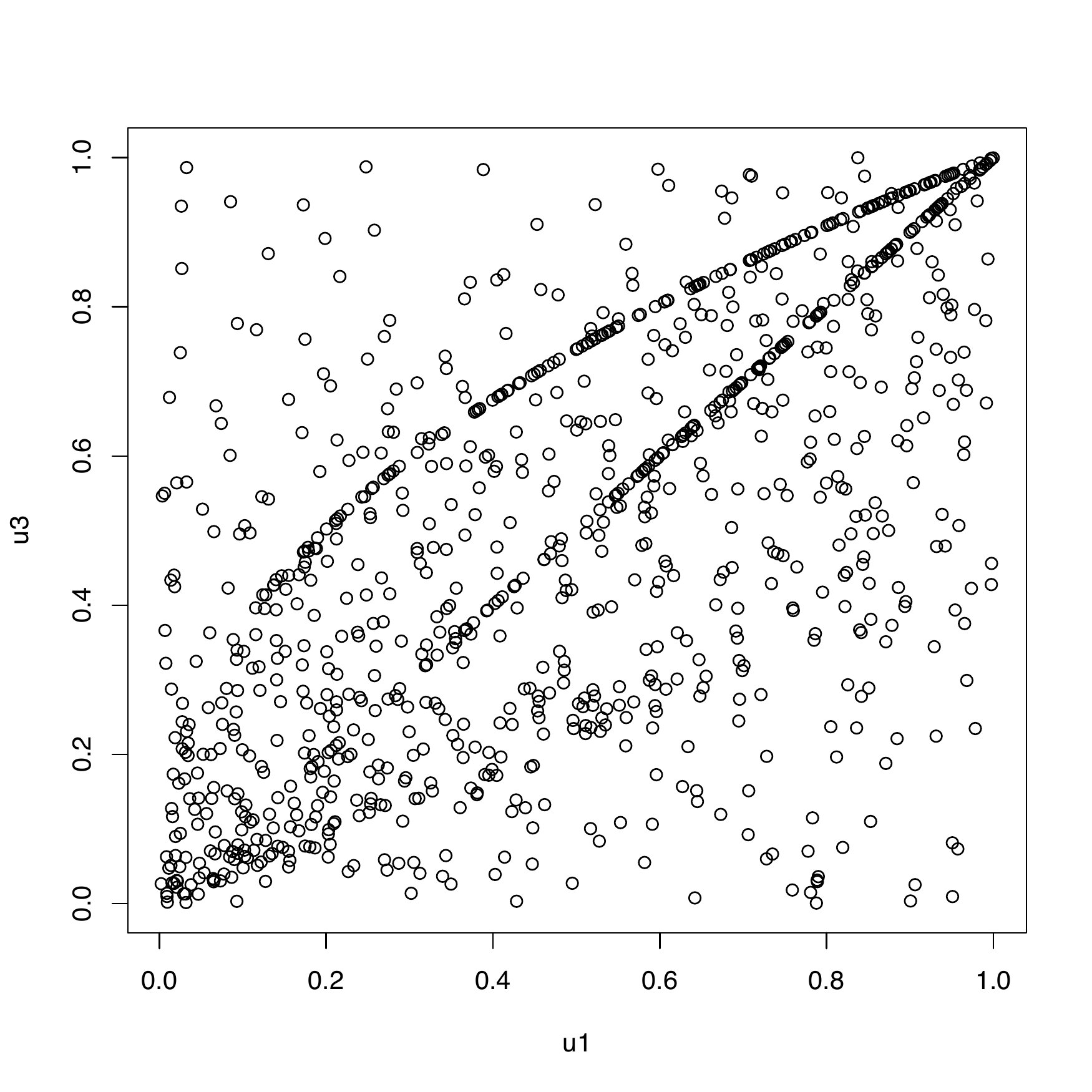}
\end{minipage}
\begin{minipage}[c]{0.5\textwidth}
\includegraphics[width=\textwidth]{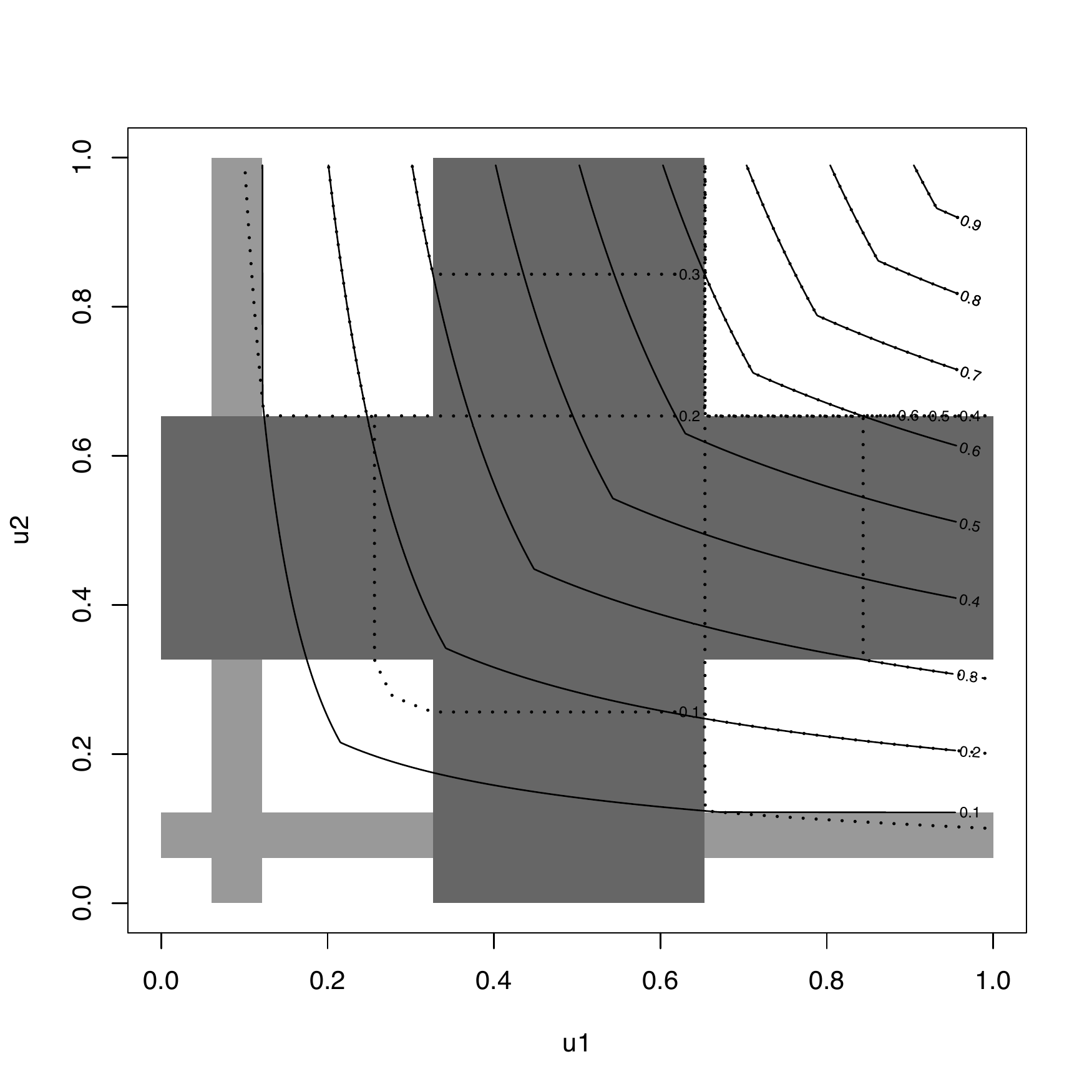}
\end{minipage}
\captionof{figure}{\textbf{left panel:} Scatterplot of the unspecified bivariate margin for a three-dimensional PCC of Cuadras-Aug\'e copulas with parameters $\alpha_{12}=0.7=\alpha_{13\vert2}$, $\alpha_{23}=0.3$.\\
\textbf{right panel:} Cdf of the bivariate conditional copulas of a trivariate MO copula with parameter $\lambda=2$ (see Appendix \ref{section:mo}). The cdf for $u_3=0.08$ is plotted solid (it is uniquely determined outside the light grey area), the dashed lines correspond to the cdf for $u_3=0.28$ (which is uniquely determined outside the dark grey area).}
\vspace{.5cm}
\label{moplots}
\end{minipage}

\subsection*{Weakening the simplifying assumption}

Whereas the case of PCCs which are not absolutely continuous will rarely be relevant in practice, there is a wide range of applications for distributions involving non-simplified PCCs.

Although the simplifying assumption sounds restrictive on the first glance, fairly general distributions can be obtained while keeping the simplified assumption for parts of the distribution.
Staying in the parametric and absolutely continuous framework, the parameters of otherwise simplified PCCs can be made depending on the values of variables that are conditioned on. In cases where there is an underlying economic assumption for how covariates should influence the dependence, this can be done in the form of parametric models. Examples for this were studied by for example \citeN{patton2006} and \shortciteN{bartram2007}, who considered models where the dependence parameter of the conditional copula depends on previous realizations of the time series, or \citeN{stoeber2011c} and \citeN{almeida2011}, who considered dependence of the parameter on an underlying Markov or AR(1) process, respectively. In particular, the Markov switching R-vine copula model of \citeN{stoeber2011c} implies that the copula of a multivariate time series at each point in time is given by a discrete mixture of simplified PCC's, which will usually be of non-simplified form as for the mixture of Normals.

When there is no a priori knowledge for how the dependence parameter should be influenced, non-parametric models as in \shortciteN{acar2011} can be applied. While this is fairly straightforward when conditioning on just one variable, it raises the question how interactions should be included when conditioning on multiple variables. For the elliptical and Archimedean distributions which we studied in Sections 3 and 4, we observe the following:

\begin{itemize}
\item For elliptical distributions, the Kendall's $\tau$ of conditional distributions does not depend on the values which are conditioned on. The distribution however can depend on these variables. Assuming zero means and correlations, and conditioning on realizations, $X_k=x_k, \ldots, X_d=x_d$ the conditional distribution will only depend on $a=\sum_{i=k}^d x_i^2$.
If the generator function $g_k(\cdot)$ is monotonely decreasing (as for the multivariate Normal or Student-t distribution, and generally if $k \geq 2$, see \cite[Section 4.9]{Joe}), this implies in particular that the conditional distribution only depends on $g_k(\x_{k:d}' \S_{k:d}^{-1} \x_{k:d})$.
For this reason, one might consider to make the dependence parameter of conditional copulas depend on the likelihood of observations that are conditioned on, for data where the distribution appears to be close to the elliptical family.
However, since the values of Kendall's $\tau$ must not be affected --- which are closely related to parameter values for most well known bivariate parametric families --- keeping the simplifying assumption will always be a close approximation in these cases.
\item For an Archimedean copula with generator function $\varphi$, we have observed that when conditioning on realizations  $X_k=x_k, \ldots, X_d=x_d$ the conditional distribution will only depend on $a=\sum_{i=k}^d \varphi^{-1}(x_i)$. In particular, this implies that it only depends on $\varphi(a)$, which is the cdf of $X_k,\ldots,X_d$, evaluated at $X_k=x_k, \ldots, X_d=x_d$. Thus, for data showing dependence behavior close to that of the Archimedean class, we recommend to consider analyzing dependence of the parameters of conditional copulas on this joint probability.
\end{itemize}
As a simple example for how the conditional copula can depend on the values of conditioning variables in the Archimedean case, let us consider the three-dimensional Frank copula. 
\begin{example}[Three-dimensional Frank copula]
\label{ex:frank}
Let $\textbf{U}_{1:3}$ have cdf $C_{Frank}(\textbf{u}_{1:3};\alpha)=\varphi_{F}\left( \varphi_F^{-1} (u_1;\alpha)+ \varphi_F^{-1} (u_2;\alpha)+ \varphi_F^{-1} (u_3;\alpha);\alpha\right)$ with LT of the logarithmic series distribution:
\begin{displaymath}
\varphi_{F}(u;\alpha)=-\alpha^{-1} \log\left( \frac{e^{-u}(e^{-\alpha}-1)+1}{\alpha}\right) \ \ \text{for} \ \alpha \in (0,\infty),\end{displaymath}
i.e., the three-dimensional Frank copula. Then, the copula corresponding to the conditional distribution of $U_1,U_2\vert U_3=u_3$ is of the Ali-Mikhail-Haq family (LT of geometric distribution)
\begin{displaymath}
C_{AMH}(u_1,u_2;(\theta(\alpha,u_3)))=\frac{u_1 u_2}{1-\theta(\alpha,u_3) (1-u_1) (1-u_2)}
\end{displaymath}
with parameter $\theta(\alpha,u_3)=1-e^{-\alpha u_3}$, thus depending on the conditioning value $u_3$.
How the strength of dependence between $U_1,U_2\vert U_3=u_3$ depends on the parameter $\alpha$ and the conditioning value $u_3$ is illustrated in Figure \ref{amh_dep}.
\end{example}
\begin{remark}[Extensions of Archimedean copulas]
As for the multivariate Student-t copula, this decomposition allows for a natural extension of the Frank copula by assuming different $\alpha$ parameters for the different bivariate copulas in the models, and similar extension arise for other Archimedean copulas. The possible range of dependence induced by this extension is illustrated in Figure \ref{frank}.
\end{remark}

\begin{minipage}[c]{\textwidth}
\vspace{-.7cm}
\begin{minipage}[c]{0.5\textwidth}
\includegraphics[width=\textwidth]{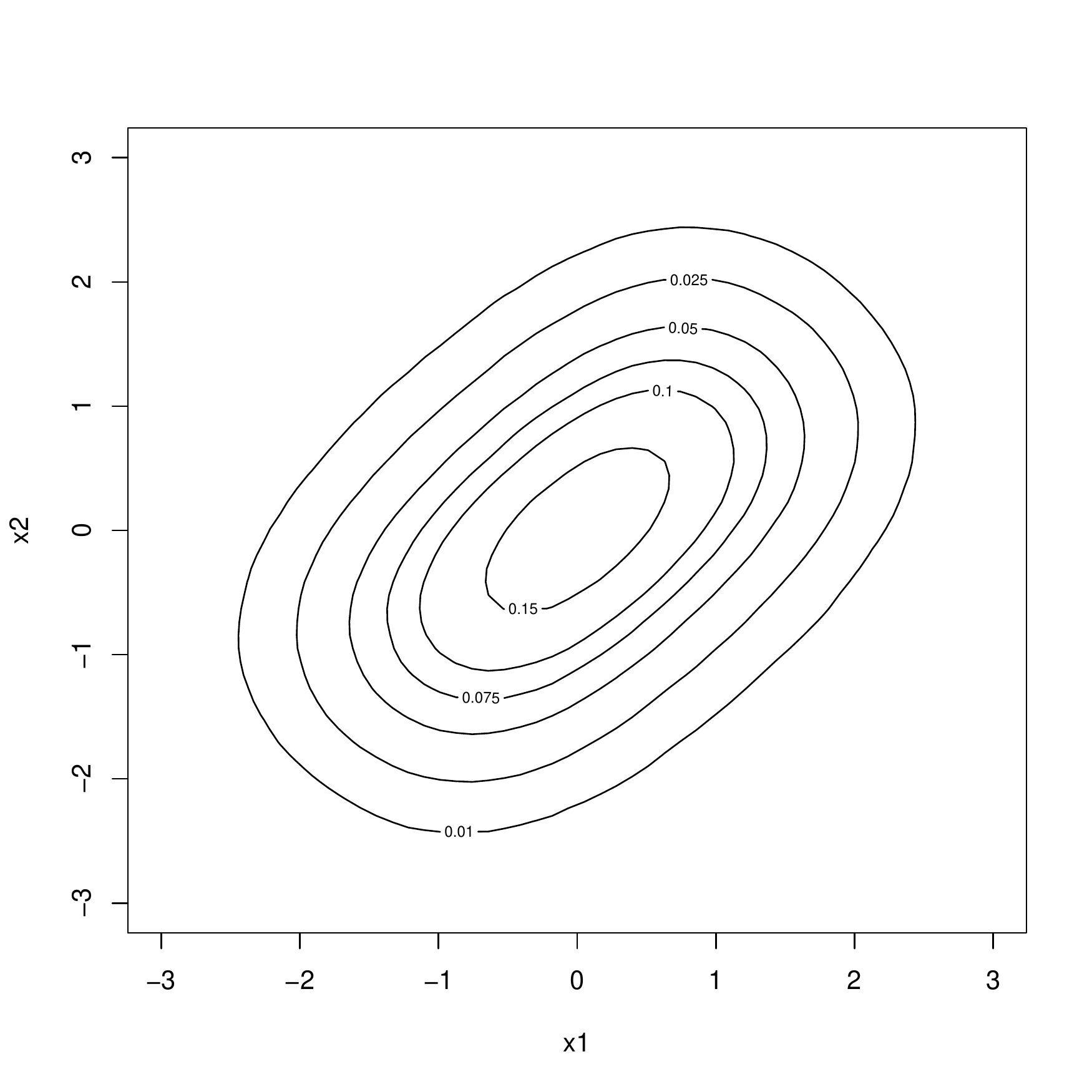}
\end{minipage}
\begin{minipage}[c]{0.5\textwidth}
\includegraphics[width=\textwidth]{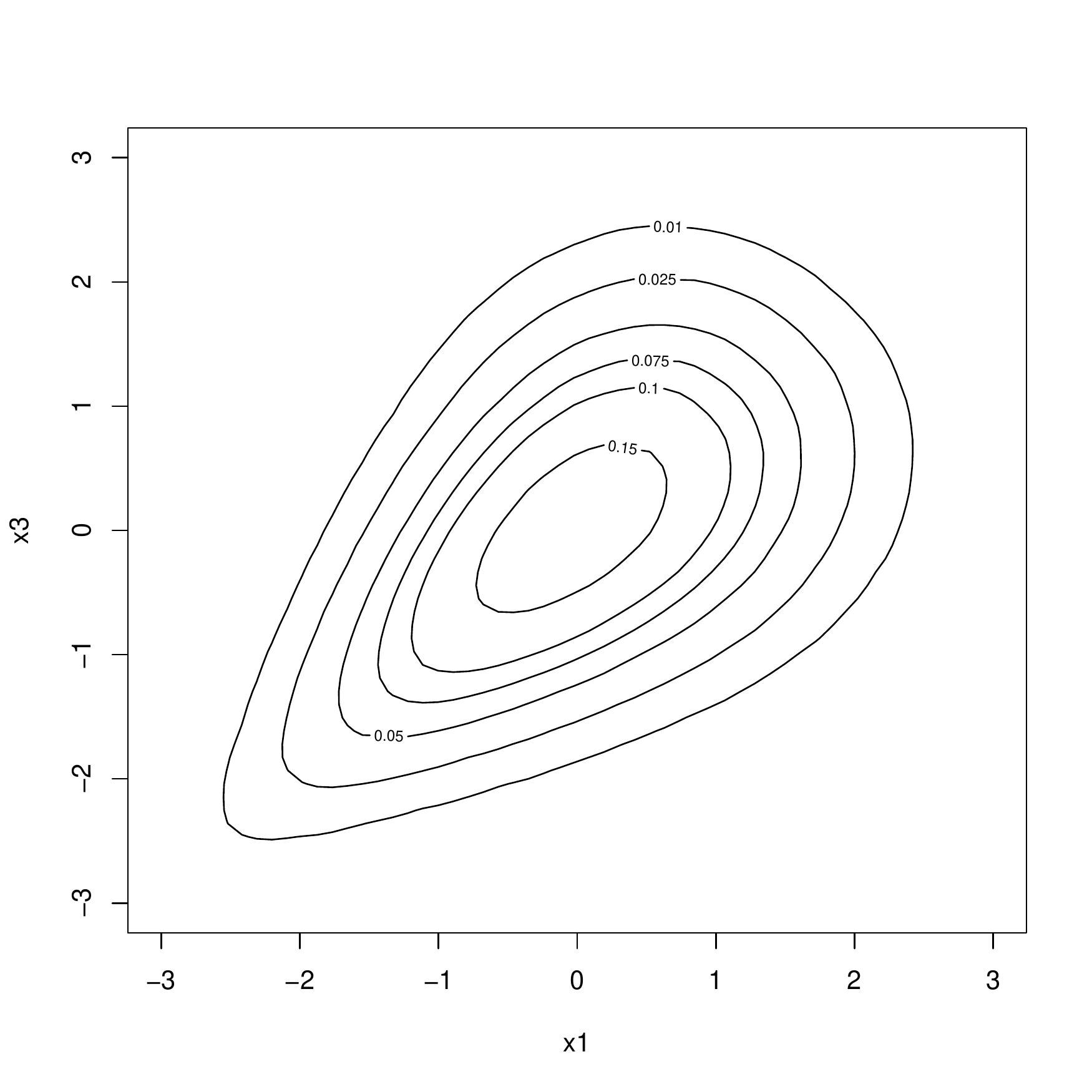}
\end{minipage}
\vspace{-.5cm}
\captionof{figure}{Density of the unspecified $13$ margin of a trivariate meta-normal distribution with extended Frank copula with $\alpha_{13}=1$ ($\tau=0.11$), $\alpha_{23}=3$ ($\tau=0.21$) and $\alpha_{12\vert 3 }=30$ (right panel). For comparison, the left panel shows the density of a bivariate meta normal distribution where the dependence is characterized by a Frank copula with $\alpha=3$ ($\tau=0.21$).}
\label{frank}
\end{minipage}

\begin{minipage}[c]{\textwidth}
\vspace{-1cm}
\begin{centering}
\includegraphics[width=\textwidth]{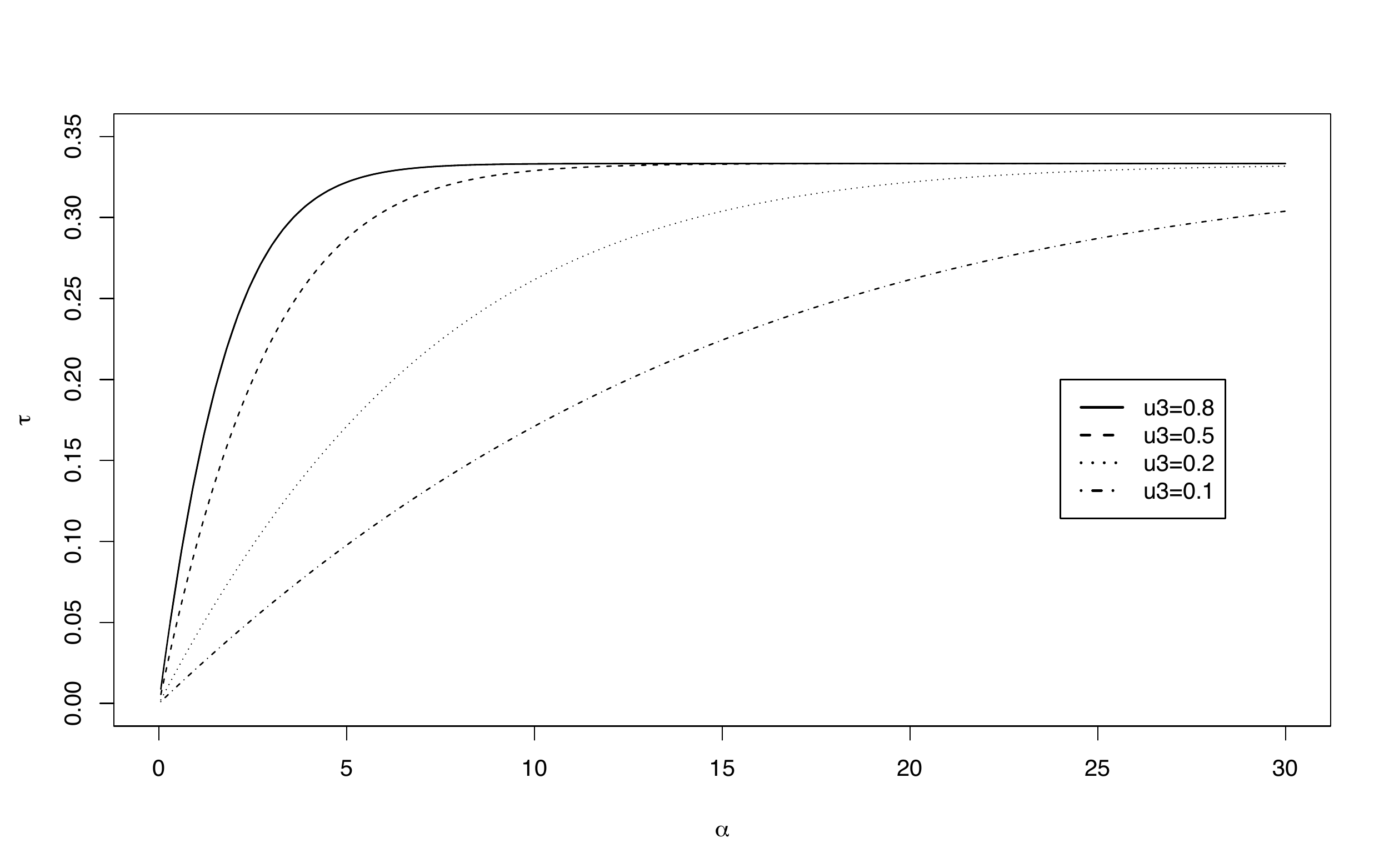}
\end{centering}
\vspace{-1.2cm}
\captionof{figure}{Kendall's $\tau$ for the AMH-copula with parameter $\theta(\alpha,u_3)=1-e^{-\alpha u_3}$ for different values of $\alpha$ and $u_3$.}
\label{amh_dep}
\end{minipage}

\section{Conclusion and summary}\label{conclusion}
While many popular statistical models are built by (sequential) conditioning, the implications on the properties of a distribution arising from this construction remained largely unknown. In addition, whether and under which assumptions popular classes of dependence models can be decomposed in this way was not investigated.

In this paper, we filled in this gap by characterizing the most common classes of copulas in terms of their decomposability as a PCC.

We showed that only the $d$-dimensional Archimedean copula based on the gamma LT or its generalization into $\mathcal{L}_d$ can be decomposed using a PCC in which the building blocks are independent of the values that are conditioned on. 
For elliptical copulas, the situation is more challenging. While the multivariate Normal and Student-t distribution are simplified PCCs, the conjecture that this is true for all elliptical distributions does not hold. The Student-t distribution even is the only non-bounded elliptical distribution in which conditioning affects only the location and scale of the resulting conditional distribution but not the correlation matrix.
Just as for the MTCJ in the Archimedean family, this distribution arises from a Gamma mixture for the square of the inverse scale parameter of a multivariate normal distribution. We have shown, that this makes it the only scale mixture of Normals which is a simplified PCC in dimension $d\geq 4$.

Whereas the simplifying assumption for PCCs is convenient, it is often too restrictive, and also the assumption of dealing with absolutely continuous PCC is sometimes too strong. We illustrated this with several examples and demonstrated that in most applications, however, simplified PCC can be extended to adapt to the situation.

While our classification results provide important insights into the properties of distributions arising from conditioning, we believe that further research in this direction remains necessary, in particular considering how well general distributions can be approximated by simplified PCCs.

Figure \ref{overview_graphic} does not include the extreme value copulas which
is another general class. Extreme value copulas are used for data that are
multivariate maxima or minima, and in this context, PCCs would not be
considered. It can be shown that extreme value copulas which do not have
independent subcomponents are not simplified PCCs; for example, given a
trivariate extreme value copula $C(u_1,u_2,u_3)$ which satisfies
$C(u_1^t,u_2^t,u_3^t)=C^t(u_1,u_2,u_3)$ for all powers with $t>0$, it can be shown that
the copula $C^*(v_1,v_2;u_3)$ obtained by conditioning on $U_3=u_3$
has lower tail form $v^\kappa \gamma(u_3,v)$ as $v\to 0$, where $\kappa\in
(1,2)$ does not depend on $u_3$ but the multiplying factor $\gamma$
is slowly varying in $v$ and depends on $u_3$.

\appendix

\section{Proofs involving Archimedean copulas - Theorem \ref{theorem:archimedean}}\label{appendix:archimedean}

Because the proof is more intuitive, we first outline the case where $\vph$ is the LT of a positive random variable. In this case, there is a representation of the copula as a mixture of powers.
The mixture distribution from which the Archimedean copula arises can be written as (c.f. \citeN[p. 86]{Joe}):
\begin{displaymath}
F(\textbf{x}_{1:d})=\int_0^\infty \prod_{j=1}^d \left[G(x_j)\right]^\alpha dF_A(\alpha)=\int_0^\infty \prod_{j=1}^d \left[G(x_j)\right]^\alpha f_A(\alpha) d\alpha,
\end{displaymath}
where $F_A$ is the cdf of a positive rv, with corresponding density $f_A$, and 
\begin{displaymath}
F(x)=\int_0^\infty [G(x)]^\alpha dF_A(\alpha) =\vph_A(-\ln G(x))
\end{displaymath}
is the common univariate cdf with $\vph_A$ being the LT of A. Without loss
of generality we can assume that $F(x)=x$ on $[0,1]$, then
$F_{1:d}(\textbf{x}_{1:d})$, $x_j \in [0,1]$, $j=1,\ldots,d$ is a copula.
Also $G(x)=\exp\{-\vph_A^{-1}(x)\}$ on $[0,1]$ is differentiable.
With $g=G'$, the marginal cdf of the last $k$ variables and its density are:
\begin{equation*}
\begin{split}
F_{(d-k+1):d}(\textbf{x}_{(d-k+1):d}) &=\int_0^\infty \left[\prod_{i=d-k+1}^d G(x_i)\right]^\alpha f_A(\alpha) d\alpha, \\
 f_{(d-k+1):d}(\textbf{x}_{(d-k+1):d}) &=\frac{\partial^k }{\partial x_{d-k+1}\cdots \partial x_d}F_{(d-k+1):d}(\textbf{x}_{(d-k+1):d})\\ &=\prod_{j=d-k+1}^d \left[\frac{g(x_j)}{G(x_j)}\right] \int_0^\infty \alpha^k \left[\prod_{i=d-k+1}^d G(x_i)\right]^\alpha f_A(\alpha) d\alpha,
\end{split}
\end{equation*}
and the conditional cdf of the first $d-k$ variables given the last $k$ is:
\begin{equation*}
\begin{split}
F_{1:(d-k)\vert(d-k+1): d}(\textbf{x}_{1:(d-k)}\vert \textbf{x}_{(d-k+1):d})
&=\frac{\partial^k F(x_1,\dots,x_d)}{\partial x_{d-k1}\cdots\partial x_d} \Bigg/ f_{(d-k+1):d}(\textbf{x}_{(d-k+1):d}) \\
&=\frac{\int_0^\infty \prod_{j=1}^{d-k} \left[G (x_j)\right]^\alpha \cdot \alpha^k  \prod_{i=d-k+1}^d \left[G(x_i)\right]^\alpha f_A(\alpha)d\alpha}{\int_0^\infty \alpha^k \prod_{i=d-k+1}^d \left[G(x_i)\right]^\alpha f_A(\alpha) d\alpha}\\ &=\int_0^\infty \prod_{j=1}^{d-k}  \left[G(x_j)\right]^\alpha \cdot f_{A^\star}(\alpha)d\alpha,
\end{split}
\end{equation*}
where
\begin{displaymath}
f_{A^\star}(\alpha)=\frac{\alpha^k  \prod_{i=d-k+1}^d \left[G(x_i)\right]^\alpha f_A(\alpha)}{\int_0^\infty \beta^k  \prod_{i=d-k+1}^d \left[G(x_i)\right]^\beta f_A(\beta) d\beta} \propto \alpha^k e^{\alpha  \sum_{i=d-k+1}^d\ln(G(x_i))} f_A(\alpha).
\end{displaymath}
In this case, $A^\star$ has the same parameter form of density as $A$ if
\begin{equation}
f_A(\alpha;\eta,\theta)=e^{-\alpha \eta} \alpha^{\theta-1} h(\alpha) / C(\eta,\theta),
\label{mixing_density}
\end{equation}
where $h$ is a positive-valued function (it is not absorbed in the $\alpha^{\theta-1}$ term only if it is a non-power function) and $C(\theta,\eta)$ is a (finite) normalizing constant. From above, $A^\star$ has the same density form with parameters $(\eta - \sum_{i=d-k+1}^d\ln G(x_i),\theta +k)$.
The conditional copula does not depend on $x_{d-k+1},\ldots,x_d$ only if $\eta$ is a rate (or reciprocal scale) parameter, since  Archimedean copulas are invariant to scale changes of the mixing distribution \cite[p. 60]{mai2012}. For $\eta$ to be a rate or inverse scale parameter of (\ref{mixing_density}) only, $h(\alpha)$ must be a power of $\alpha$. Hence $f_A$ is a gamma density, and $F(\textbf{x}_{1:d})$ is a MTCJ copula.

For the general case, where $\vph \in \mathcal{L}_d$ is not necessarily a LT, we prove the result by construction a functional equation. Let
$$C(u_1,\ldots,u_d)=
 \vph\Bigl(\sum_{j=1}^d \vph^{-1}(u_j)\Bigr)   $$
be an Archimedean copula. 
Let $(U_1,\ldots,U_d)$ be a random vector with this distribution.
Suppose $\vph$ has support $[0,s_0)$ where $s_0=\inf\{s, \varphi(s)=0 \}$ is infinite for a Laplace
transform, but could be finite for $\vph\in{\cal L}_d$ which is not
a Laplace transform. The case of finite support implies that
$\vph(s)=0$ for $s\ge s_0$.  
Let $C_{1\cdots d-1|d}(u_1,\ldots,u_{d-1}|u_d)=
{\p C(u_1,\ldots,u_d)/\p u_d}$ be the conditional distribution given
$U_d=u_d$. We will show now that the copula for this is another Archimedean copula,
say based on $\psi$, where $\psi\in{\cal L}_{d-1}$.
By differentiation, with $h=-\vph'$ and $a=\vph^{-1}(u_d)\in[0,s_0)$ 
with $0< u_d\le 1$,
 $$ C_{1\cdots d-1|d}(u_1,\ldots,u_{d-1}|u_d)= 
  {h\Bigl(\sum_{j=1}^d \vph^{-1}(u_j)\Bigr) \Bigm/ h(a) },$$
with $j$th ($1\le j\le d-1$) margin $F_j(u_j|u_d)=
h( \vph^{-1}(u_j)+ a)/h(a)=:v_j$.
Note that $h$ is monotone increasing by definition of ${\cal L}_d$.
Hence $\vph^{-1}(u_j)=h^{-1}(v_jh(a))-a$ for $1\le j\le d-1$
and the copula of the conditional distribution of $U_1,\ldots, U_{d-1}$ given $U_d=u_d$ is:
 $$ C^*(v_1,\ldots,v_{d-1};a)=h\Bigl(\sum_{j=1}^{d-1} h^{-1}(v_jh(a))
    -(d-2)a\Bigr) \Bigm/ h(a) $$
    
Defining $s=\psi^{-1}(v;a)=h^{-1}(vh(a))-a$ and $v=\psi(s;a)=h(s+a)/h(a)$,
this is a Archimedean copula
$$\psi(\psi^{-1}(v_1;a)+\cdots+\psi^{-1}(v_{d-1};a);a)$$ with 
generator function $\psi(\cdot,a)$.
As $u_d\to1$, $a=\vph^{-1}(u_d)\to0$, and $h(0)=-\vph'(0)$ can
be positive or infinite but not 0, by the definition of $\varphi$.

Consider first the case where $h(0)$ is finite and positive, and $h'(0)$
is finite.
The copula of the conditional distribution does not depend on $u_d$ or $a$
if and only if there is a continuous differentiable scale function $\gamma(a)>0$ 
such that $\gamma(0)=1$ and
 $$\psi(s;a)=h(s+a)/h(a)=\psi(s\gamma(a);0)=h(s\gamma(a))/h(0); \quad
  0\le s<s_0.$$
Writing the above functional equation in $h$ as
$h(s+a)\,h(0)=h(a)\,h(s\gamma(a))$ and differentiating with respect to $a$ yields
 $$h'(s+a)\,h(0)=h'(a)\,h(s\gamma(a))+h(a)\,h'(s\gamma(a))\,s\gamma'(a).$$
With $a=0$ if follows that
 $$h'(s)\,h(0)=h'(0)\,h(s)+h(0)\,h'(s)\,s\gamma'(0),$$
or
  $$h'(0)\,h(s) =  h(0)[1-s\gamma'(0) ] h'(s)$$
This has solution $h(s)=h(0)[1-s\gamma'(0)]^{\alpha}$ where
$\alpha=-h'(0)/[h(0)\gamma'(0)]$. Since $h=-\psi'$ must be decreasing,
there are 2 possibilities 
\begin{itemize}
\item[(i)] $s_0=\oo$, $\gamma'(0)<0$, $\alpha<0$, or
\item[(ii)] $\gamma'(0)>0$, $s_0=1/\gamma'(0)$, $\alpha>0$.
\end{itemize}
By integrating $h$ over $s$ we obtain $\vph(s)=(1-s\gamma'(0))_+^{1+\alpha}$, since $\psi(0)=0$.
In case (i), we must have $\alpha< -1$ and in case (ii), $1+\alpha\ge d-1$
in order for $\psi\in{\cal L}_d$ (see also \citeN[pp. pages 157--158]{Joe}). In case (i) the obtained generating function support on $[0,\infty)$ and corresponds to the "standard" MTCJ copula, whereas case (ii), with bounded $\vph$ yields the extended MTCJ copula.

If $h'(0)$ or $h(0)$ is infinite, the above is modified as follows.
Let $0<\eps<s_0$.
There is a continuous differentiable scale function $\gamma(a)>0$ 
such that $\gamma(\eps)=1$ and
 $$\psi(s;a)=h(a+s)/h(a)=\psi(s\gamma(a);\eps)=h(s\gamma(a)+\eps)/h(\eps).$$ 
Cross-multiplying and differentiating the above with respect to $a$, and then setting $a$ to $\eps$ yields
  $$h'(\eps)\,h(s+\eps) =  h(\eps)[1-s\gamma'(\eps) ] h'(s+\eps),
  \quad  0<s<s_0-\eps.$$
This has solution $h(s+\eps)=h(\eps)[1-s\gamma'(\eps)]^{\alpha}$ where
$\alpha=-h'(\eps)/[h(\eps)\gamma'(\eps)]$ so 
that $\psi(s;\eps)=h(s+\eps)/h(\eps)=[1-s\gamma'(\eps)]^{\alpha}$.
The conclusion is the same as above, because after integrating $h$ to get
$\vph$, one would conclude that this leads to $\vph'(0)$ and $\vph''(0)$
being finite. \hfill $\square$

As an illustration for how the copulas of conditional distributions are derived when the conditioning is on more than one variable, let us consider the case where $d\geq 4$ and 
$$C(u_1,\ldots,u_d)=\vph\left(\sum_{j=1}^d\vph^{-1}(u_j)\right)$$ is an Archimedean copula.
Let $(U_1,\ldots,U_d)$ be a random vector with this distribution and let $a_j=\varphi^{-1}(u_j)$ for $j=d-1$ and $j=d$.
Then, the conditional distribution of $(U_1,\ldots,U_{d-2})$ given $U_{d-1}=u_{d-1}, U_d=u_d$ is
\begin{equation*}
\begin{split}
F_{1:(d-2)\vert d-1,d}(u_1\ldots,u_{d-2}\vert u_{d-1},u_d)&= \frac{\partial^2 C(u_1,\ldots,u_d) / \partial u_{d-1} u_d}{\vph''(a_{d-1}+a_d)/ [\vph'(a_{d-1}) \vph'(a_d)]}\\ &=\frac{\vph''\left( \sum_{j=1}^{d-2} \vph^{-1}(u_j) + a_{d-1}+ a_d \right)}{\vph''(a_{d-1}+a_d)}.
\end{split}
\end{equation*}
The copula corresponding to this distribution is again Archimedean, based on a generator $\psi$. By differentiation, with $h=\vph''$, $a=a_{d-1}+a_d$, $F_{1:(d-2)\vert d-1,d}(\cdot \vert u_{d-1},u_d)$ has $j$th $(a\leq j \leq d-2)$ margin $F_j(u_j\vert u_{d-1},u_d)=h(\vph^{-1}(u_j) + a) / h(a) =: v_j$. Note that $h$ is monotone decreasing because $\vph$ is the generator of a $d$-dimensional Archimedean copula. Hence, $\vph^{-1}(u_j)=h^{-1}(v_j h(a)) -a$ for $1\leq j \leq d-2$ and the copula of the conditional distribution can be expressed as 
\begin{equation*}
C_{1:(d-2)\vert d-1,d}(v_1,\ldots,v_{d-2};a)=h\left(\sum_{j=1}^{d-2} h^{-1}(v_j h(a)) - (d-3) a\right) \Big/ h(a).
\end{equation*}
This is an Archimedean copula $\psi(\psi^{-1}(v_1;a)+\cdots + \psi^{-1}(v_{d-2};a);a)$ with $s=\psi^{-1}(v;a) = h^{-1}(vh(a))-a$ and $v=\psi(s;a) = h(s+a)/h(a)$.
The same pattern extends to conditional distributions of Archimedean copulas with three or more conditioning variables.

\section{Proofs involving elliptical copulas}

\subsection{Theorem 4.1}\label{appendix:student1}
In the remainder, we will use the following notation. A $d$-dimensional Student-t distribution with mean vector $\textbf{0}$, correlation matrix $R$ and degrees of freedom $\nu$ is denoted as $t_d(\textbf{0},R,\nu)$. Its pdf is $f_{t,d}(\cdot; R,\nu)$ and we write $F_{t,d}(\cdot; R,\nu)$ for the cdf.
%For this proof, we will use the following notation.
%\renewcommand{\arraystretch}{1.1}
%\begin{center}
%\begin{tabular}{c l}
% \begin{tabular}{ccccc}
%Student-t &\ \ \ \ \ pdf\ \ \ \ \ &\ \ \ \ \ cdf\ \ \ \ \ &\ \ \ \ mean \ \ \ \ & \ \ \ \ cor\ \ \ \\
%\hline
%$t_d(\textbf{0},R,\nu)$ & $f_{t,d}(\cdot; R,\nu)$ & $F_{t,d}(\cdot; R,\nu)$ & $\textbf{0}$ & R \\
%\hline
%\multicolumn{5}{c}{$\textbf{X}=\begin{pmatrix} \textbf{X}_A \\ \textbf{X}_B \end{pmatrix}\sim t_d(\textbf{0},R,\nu) \ \ \ \ \ R=\begin{pmatrix} R_A & R_{AB} \\ R_{AB}^t & R_B \end{pmatrix}$}\\
%\hline
%\end{tabular}
%%\caption{Notation for multivariate distributions.}
%\label{table_notation_multivariate_distributions}
%%\end{tabular}
%\end{center}

Let us consider a $d$-dimensional random vector $\textbf{X}=(\textbf{X}_A,\textbf{X}_B)=(X_1,X_2,\textbf{X}_B)$, with $A=\{1,2\}$, distributed according to a multivariate Student-t distribution with $\nu$ degrees of freedom, mean $\bmu=(\mu_1,\dots,\mu_n)^T$ and scale matrix 
\begin{displaymath}
R=(R_{i,j})_{i,j=1,\dots,d}=
%\begin{pmatrix}
%R_{11} & R_{12} & R_{1C}\\
%R_{21} & R_{22} & R_{2C}\\
%R_{1C}^T & R_{2C}^T & R_C\\
%\end{pmatrix}.
\begin{pmatrix}
R_A & R_{AB}\\
R_{AB}^T & R_B
\end{pmatrix},
\text{where } R_{AB}=\begin{pmatrix}
R_{1B} \\
R_{2B}
\end{pmatrix}, R_A=\begin{pmatrix}
R_{11} & R_{12} \\
R_{21} & R_{22}\\
\end{pmatrix}.
\end{displaymath}
Let us define
\begin{equation*}
\begin{split}
&V_{A\vert B}:= \begin{pmatrix} R_{11} & R_{12}\\ R_{21} & R_{22}\end{pmatrix}-\begin{pmatrix}R_{1B} \\ R_{2B} \end{pmatrix} R_B^{-1} \begin{pmatrix} R_{1B}^T & R_{2B}^T \end{pmatrix}=:\begin{pmatrix}V_{1\vert B} &  V_{{A\vert B}_{12}} \\ V_{{A\vert B}_{21}} & V_{2\vert B} \end{pmatrix},\\
&R_{A\vert B}:=diag(V_{A\vert B})^{-\frac{1}{2}}\ V_{A\vert B}\ diag(V_{A\vert B})^{-\frac{1}{2}}, 
\gamma(\textbf{x}_B):=\sqrt{\frac{1+(1/\nu)\textbf{x}_B^T R_B^{-1} \textbf{x}_B}{(\nu+d-2)/\nu}}, \\
\end{split}
\end{equation*}
then we have for the conditional distribution of $\textbf{X}_A$ given $\textbf{X}_B=\textbf{x}_B$:
\begin{equation*}
F_{A\vert B}(\textbf{x}_A\vert \textbf{x}_B)= F_{t,2}\left(\frac{x_1-\mu_{1\vert B}(\textbf{x}_B)}{\sqrt{{V}_{1\vert B}}\cdot \gamma(\textbf{x}_B)}, \frac{x_2-\mu_{2\vert B}(\textbf{x}_B)}{\sqrt{{V_{2\vert B}}}\cdot \gamma(\textbf{x}_B)};R_{A\vert B},\nu+d-2 \right),
\end{equation*}
c.f. \shortciteN[Lemma 2.2]{nikoloulopoulos2009}.
 Taking $x_2 \rightarrow \infty$ ($x_1 \rightarrow \infty$) yields
\begin{equation*}
\begin{split}
&F_{1\vert B}(x_1\vert \textbf{x}_B)=F_{t,1}\left(\frac{x_1-\mu_{1\vert B}(\textbf{x}_B)}{\sqrt{{V}_{1\vert B}}\cdot \gamma(\textbf{x}_B)};\nu+d-2 \right)\\
&F_{2 \vert B}(x_2\vert \textbf{x}_B)=F_{t,1}\left(\frac{x_2-\mu_{2\vert B}(\textbf{x}_B)}{\sqrt{{V}_{2\vert B}}\cdot \gamma(\textbf{x}_B)};\nu+d-2 \right),
\end{split}
\end{equation*}
so that we can now determine the corresponding copula:
\begin{equation*}
\begin{split}
C_{1,2\vert 3:d}(u_1,u_2)&=C_{1,2\vert 3:d}(u_1,u_2\vert \textbf{x}_B)=F_{12\vert 3:d}\left(F_{1\vert 3:d}^{-1}(u_1\vert \textbf{x}_B),F_{2\vert 3:d}^{-1}(u_2\vert \textbf{x}_B) \vert \textbf{x}_B\right)\\
&=F_{t,2}\left({F_{t,1}}^{-1}(u_1;\nu+d-2),{F_{t,1}}^{-1}(u_2;\nu+d-2);R_{A\vert B},\nu+d-2 \right),
\end{split}
\end{equation*}
where the additive constants and scaling factors cancel. This is a bivariate Student-t copula with $\nu+d-2$ degrees of freedom and correlation matrix $R_{A\vert B}$ and does not depend on $\textbf{x}_B$ anymore.\hfill $\square$

\subsection{Theorem 4.3}\label{appendix:student2}
The proof is similar to that for Archimedean copulas in that the same functional equation can be obtained.
Without loss of generality, let us consider the case of a $d$-dimensional elliptical distribution with zero means and zero correlations. In this case, the density of the distribution is given by
\begin{equation*}
\begin{split}
&f_{1:d}(\textbf{x}_{1:d})=g_d(x_1^2+\ldots+x_d^2), \text{with marginal density }\\ &f_{(d-k+1):d}(\textbf{x}_{(d-k+1):d})=g_k(x_{d-k+1}^2+\ldots+x_d^2).
\end{split}
\end{equation*}
For $f_{1:(d-k)\vert (d-k+1):d}(\textbf{x}_{1:(d-k)}\vert \textbf{x}_{(d-k):d})$ we obtain
\begin{displaymath}
f_{1:(d-k)\vert (d-k+1):d}(\textbf{x}_{1:(d-k)}\vert \textbf{x}_{(d-k):d})=\frac{g_d(x_1^2+\ldots+x_d^2)}{g_k(x_{d-k+1}^2+\ldots+x_d^2)}.
\end{displaymath}
For this distribution to be in the same location-scale family irrespective of the values of $x_{d-k+1},\ldots, x_d$, we must have that for given $\sum_{i=d-k+1}^d x_i^2\neq \sum_{i=d-k+1}^d {x_i^\star}^2$ there exists $\gamma(\textbf{x}_{(d-k+1):d},\textbf{x}_{(d-k+1):d}^\star)$ such that
\begin{displaymath}
f_{1:(d-k)\vert (d-k+1):d}(\gamma \cdot \textbf{x}_{1:(d-k)} \vert \textbf{x}_{(d-k+1):d}^\star) \propto f_{1:(d-k)\vert (d-k+1):d}(\textbf{x}_{1:(d-k)} \vert \textbf{x}_{(d-k+1):d}).
\end{displaymath}
With $\textbf{x}_{(d-k+1):d}^\star=\mathbf{0}$, $a=\sum_{i=d-k+1}^d x_i^2$, $t=\sum_{i=1}^{d-k} x_i^2$ and $\delta(a)=\frac{1}{\gamma(\textbf{x}_{(d-k+1):d},0)}$ this implies that
\begin{equation}\label{eq:t43-1}
g_3(t+a)=\xi(a) \cdot g_3\left(\frac{t}{\delta(a)}\right).
\end{equation}
where $\xi(a)$ equals $g_1(a)$ times a constant depending on $a$ and $\xi(\cdot)$, $\delta(\cdot)$ are differentiable scale functions.
Since $g_3(0)=0$ if and only if $g_3(a)=0$ for all values of $a$, we must have $g_3(0) > 0$.
Using $t=0$ in Equation \ref{eq:t43-1}, we conclude that $g_3(0)$ is finite. Thus, we can define $h(t):=\frac{g_3(t)}{g_3(0)}$ and obtain from Equation \ref{eq:t43-1} that
\begin{displaymath}
h(t+a)=h(a) \cdot h\left(\frac{t}{\delta(a)}\right).
\end{displaymath}
Using that $\delta(0)=1$ by the definition of $\delta$,
\begin{equation*}
\begin{split}
\frac{d}{da}: \ & h'(t+a)=h'(a) \cdot h\left(\frac{t}{\delta(a)}\right) + h(a) \cdot h'\left(\frac{t}{\delta(a)}\right) \cdot \left(\frac{-t \delta'(a)}{\delta(a)^2}\right) \\
a=0: \ & h'(t) = h'(0) \cdot h(t) + h'(t) \cdot (-t \cdot \delta'(0) ) .
\end{split}
\end{equation*}
In other words, the function $h$ must fulfill the differential equation
\begin{displaymath}
(1+\beta t) h'(t) = \alpha h(t),
\end{displaymath}
where $\alpha=h'(0)$, $\beta=\delta'(0)$. From this, we obtain for $\beta>0$ that $h(t)=(1+\beta t)^\frac{\alpha}{\beta}$ which corresponds to the elliptical generator of a Pearson type VII (scaled Student-t) distribution. For $h$ to yield a well defined density in $d$ dimensions, $\frac{\alpha}{\beta}$ must be given in the form $\frac{\alpha}{\beta}=-(\nu +d)/2$, $\nu >0$, to ensure integrability with respect to $t^{d/2-1}$. 

For $\beta<0$, $h(t)=(1-\beta t)_+^\frac{\alpha}{\beta}$, which leads to a well-defined density for $\frac{\alpha}{\beta} > -1$ and is differentiable and thus a valid solution for $\frac{\alpha}{\beta} > 1$. 

By integration, we obtain that the generator function for lower-dimensional margins $g_{d-k}(t)$ is proportional to $(1-\beta t)_+^{\frac{\alpha}{\beta}+\frac{k}{2}}$. Ergo, also the conditional distributions of the lower-dimensional margins remain within the same location-scale family for all values of the conditioning variables.
\hfill $\square$
\subsection{Theorem 4.4}\label{appendix:student3}
A general scale mixture of Normals in dimension $d$ can be written as $(X_1,\ldots,X_d)=(Z_1,\ldots,Z_d)/\rt{W}$ where
$W$ is a random variable on $(0,\oo)$ with density $f_W$,
and $(Z_1,\ldots,Z_d)$ is multivariate Gaussian with zero mean vector
and covariance matrix $\S$. Without loss of generality, we can assume that
all diagonal entries of $\S$ are $1$, i.e., $\S$ is a correlation matrix.
This implies for the $d$-variate generator $g_d$ that
 $$|\S|^{-1/2}g_{d}(\x'\S^{-1}\x)
  = (2\pi)^{-d/2} |\S|^{-1/2} \int_0^\oo w^{d/2} \exp\{ -\half w\x'\S^{-1}\x \}
   f_W(w)\,dw.$$
Similarly, if $\S_k$ is the leading $k\times k$ matrix of $\S$, 
and $\x_k=(x_1,\ldots,x_k)'$, then by marginalizing  out the last
$d-k$ components we get
  $$|\S_k|^{-1/2}g_k(\x'_k\S_k^{-1}\x_k) = (2\pi)^{-k/2} |\S_k|^{-1/2} \int_0^\oo w^{k/2} 
  \exp\{ -\half w\x'_k\S_k^{-1}\x_k \} f_W(w)\,dw.$$
Hence, the univariate margin is
 $$f_1(x_1)=g_1(x_1^2)=(2\pi)^{-1/2} \int_0^\oo w^{1/2} \exp\{ -\half wx_1^2\}
  f_W(w)\,dw.$$
In the above, the density $f_W$ could be replaced with $dF_W$ in a Stieltjes integral. However, this would not affect the remainder of this proof; we omit it for notational convenience.
Note that
  $$g_1(0)=(2\pi)^{-1/2}\E[W^{1/2}], \quad
  g_d(0)=(2\pi)^{-d/2}\E[W^{d/2}].$$
Let $G$ be the cdf corresponding to the marginal generator $g_1$.
Then, the copula density for $g_d$ is
 \begin{equation}\label{eq:elliptical_copula_density}
 c(u_1,\ldots,u_d;g_d,\S)=|\S|^{-1/2}{ g_d(\x'\S^{-1}\x) \over
   g_1((G^{-1}(u_1))^2) \cdots g_1((G^{-1}(u_d))^2) } 
\end{equation}
with $x_j=G^{-1}(u_j)$. From this general form of the density, we will obtain two equations for the moments of the mixing variable $W$ which will lead to necessary conditions on the conditional distributions of simplified PCCs in the elliptical class. Directly from (\ref{eq:elliptical_copula_density}) we get that
  \begin{equation}\label{E1}
  c(0.5,\ldots,0.5;g_d)= |\S|^{-1/2} {g_d(0) \over g_1^d(0) }
  = |\S|^{-1/2} {\E[W^{d/2}] \over  \E^d[W^{1/2}] } 
 \end{equation}
This means that if $W_1,W_2$ are two different mixing variables, then
the copula densities with fixed $\S$ are different unless the necessary condition of
  $$ {\E[W^{d/2}_1] \over  \E^d[W^{1/2}_1] }
  = {\E[W^{d/2}_2] \over  \E^d[W^{1/2}_2] }$$
holds. 
For the second equation, let us consider
 \begin{equation}\label{eq:elliptical_copula_derivative}
 c(0.5,\ldots,0.5,u;g_d)=|\S|^{-1/2} {g_d(\alpha x^2) \over g_1^{d-1}(0) g_1(x^2 )} 
%  = { (2\pi)^{-d/2}|\S|^{-1/2} 
  % \int_0^\oo w^{d/2}\exp\{-\half wax^2\} f_W(w)\,dw \over
  % g_1^{d-1}(0) (2\pi)^{-1/2} \int_0^\oo w^{1/2}\exp\{-\half wx^2\} f_W(w)\,dw
  %}
\end{equation}
where $\alpha\ge 1$ is the $(d,d)$ element of $\Sigma^{-1}$ and $x=G^{-1}(u)$. Note that $\alpha$ is not a scaling factor of the marginal distribution but a function of the correlation matrix. In particular, we can obtain all $\alpha \geq 1$ from correlation matrices 
\begin{displaymath}
\Sigma_\alpha = \begin{pmatrix}
\mathbbmss{1}_{d-2} & \textbf{0} & \textbf{0}\\
\textbf{0} & 1 & 1-\frac{1}{\alpha}\\
\textbf{0} & 1-\frac{1}{\alpha} & 1
\end{pmatrix},
\end{displaymath}
where $\mathbbmss{1}_{d-2}$ is the $(d-2)\times(d-2)$ identity matrix.
Using $\partial x/\partial u=1/g_1(x^2)$, we determine the first derivative of (\ref{eq:elliptical_copula_derivative}) as
  $$|\S|^{-1/2} \left[ { \alpha g'_d(\alpha x^2)\over g_1^{d-1}(0) g_1(x^2) }
    - {g_d(\alpha x^2) g'_1(x^2) \over g_1^{d-1}(0) g_1^2(x^2) }
  \right] \cdot {2x \over g_1(x^2)}. $$
After taking the second derivative with respect to $u$ and then setting $x=0$ $(u=0.5)$,
all terms are 0 except
\begin{equation}\label{F1}
\begin{split}
\lim_{x\to 0}& |\S|^{-1/2} \left[ { ag'_d(\alpha x^2)\over g_1^{d-1}(0) g_1(x^2) }
    - {g_d(\alpha x^2) g'_1(x^2) \over g_1^{d-1}(0) g_1^2(x^2) }
  \right] \cdot {2 \over g_1^2(x^2)}\\
  &= 2|\S|^{-1/2} \left[ {\alpha g'_d(0) \over g_1^{d+2}(0) }
   - { g_d(0) g'_1(0) \over g_1^{d+3}(0)} \right]\\
&= (2\pi)^{-1}|\S|^{-1/2} 
  \left[ -{\alpha\E(W^{(d+2)/2}) \over \E^{d+2}(W^{1/2})}
   + {\E(W^{d/2})\E(W^{3/2}) \over \E^{d+3}(W^{1/2})}\right].
   \end{split}
\end{equation}

Let us now consider the analogon for conditional densities.
Let $f_{1\cdots d}(\x)$ be the density of $(X_1,\ldots,X_d)$ and
let $f_1$ be the density of $X_1$ or any $X_j$. Let
 $$f_{1\cdots d-1|d}(x_1,\ldots,x_{d-1}|x_d)=f_{1\cdots d}(\x)/f_1(x_d).$$

For this, we decompose $\x'\S^{-1}\x$ as $(\x^*)'\S^{-1}_{11\cdot 2}\x^*$
where $\x^*=(x_1-\s_{1d}x_d,\ldots,x_{d-1}-\s_{d-1,d}x_d)'$ 
and $\S_{11\cdot 2}$ is the conditional covariance matrix of
$(Z_1,\ldots,Z_{d-1})$ given $Z_d$.
Writing the conditional densities in mixture form,
\begin{equation*}
\begin{split}
f_{1\cdots d}(\x)&=\int_0^\oo w^{d/2}\phi_d(\x w;R)\,f_W(w)\,dw \\
  &= (2\pi)^{-d/2} |\S|^{-1/2} \int_0^\oo w^{d/2} \exp\{ -\half w\x'\S^{-1}\x \}
   f_W(w)\,dw\\
   &= (2\pi)^{-d/2} |\S_{11\cdot2}|^{-1/2} \int_0^\oo w^{d/2} \exp\{ -\half 
  w{\x^*}'
  \S^{-1}_{11\cdot 2}\x^*\} \exp\{-\half wx_d^2\} f_W(w)\,dw,
   \end{split}
   \end{equation*}
and
 $$f_1(x_d) = (2\pi)^{-1/2}\int_0^\oo w^{1/2} \exp\{-\half wx_d^2\}f_W(w)\,dw$$
so that
 \begin{equation}\label{E2}
 \begin{split}
 &f_{1\cdots d-1|d}(x_1,\ldots,x_{d-1}|x_d)= \\ &a(x_d)
  (2\pi)^{-(d-1)/2} |\S_{11\cdot2}|^{-1/2} \int_0^\oo w^{(d-1)/2} \exp\{ 
  -\half w{\x^*}'
  \S^{-1}_{11\cdot 2}\x^*\} f_{W^*}(w)\,dw,
 \end{split} 
 \end{equation}
is a scale mixture with mixing density
$f_{W^*}(w;x_d) = w^{1/2} \exp\{-\half wx_d^2\} f_W(w)/a(x_d)$,
where $a(x_d)$ is a normalizing constant. We denote the random variable with this density by $W^*(x_d)$.

For $d\ge 3$, Equations (\ref{E1}) and (\ref{E2}) imply that a necessary condition for the copula corresponding to the distribution of $X_1,\ldots,X_{d-1}$ given $X_d=x_d$ to be independent of $x_d$ is  that
 \begin{equation}\label{E3}
 \begin{split}
 { \int_0^\oo w^{d/2} \exp\{-\half wx_d^2\}f_W(w)\,dw \cdot
  \Big(\int_0^\oo w^{1/2} \exp\{-\half wx_d^2\}f_W(w)\,dw\Bigr)^{d-2} \over 
 \Bigl(\int_0^\oo w \exp\{-\half wx_d^2\}f_W(w)\,dw\Bigr)^{d-1}}\\ = {\E[\{W^*(x_d)\}^{(d-1)/2}] \over  \E^{d-1}[{W^*}^{1/2}(x_d)] }
 \end{split}
 \end{equation}
is a constant over $x_d$.
Similarly, we obtain from (\ref{F1}) that 
\begin{equation}\label{F2}
|\S|^{-1/2} 
  \left[ -{\alpha\E(W^*(x_d)^{(d+1)/2}) \over \E^{d+1}(W^*(x_d)^{1/2})}
   + {\E(W^*(x_d)^{(d-1)/2})\E(W^*(x_d)^{3/2}) \over \E^{d+2}(W^*(x_d)^{1/2})}\right],
\end{equation}
must be equal to a constant $\beta(\alpha)$. To rewrite these equations, let $t=\half x_d^2\ge 0$ and let $V$ be a random variable with
density $f_V(v)$ proportional to $v^{1/2}f_W(v)$. 
let $V_t$ be random variable with density proportional to $e^{-tv}f_V(v)$;
this is an Laplace transform tilt of the density  of $V$ with normalizing
constant $\varphi_V(t)$, the LT of $V$ at $t$. Note that
$V_t$ has finite positive integer moments for $t>0$ and $V_0=V$.

Then, (\ref{E3}) can be rewritten as
 \begin{equation} \label{E4}
 \begin{split}{ \int_0^\oo v^{(d-1)/2} \exp\{-vt\}f_V(v)\,dv \cdot
  \Big(\int_0^\oo  \exp\{-vt\}f_V(v)\,dv\Bigr)^{d-2} \over 
 \Bigl(\int_0^\oo v^{1/2} \exp\{-vt\}f_V(v)\,dv\Bigr)^{d-1}}
 = { \E[V_t^{ (d-1)/2}] \over \E^{d-1}[V_t^{1/2}] },
  \end{split}
  \end{equation}
which must be constant over $t\ge 0$ while (\ref{F2}) leads to
\begin{equation}\label{F3}
|\S|^{-1/2} 
  \left[ -{\alpha\E(V_t^{(d+1)/2}) \over \E^{d+1}(V_t^{1/2})}
   + {\E(V_t^{(d-1)/2})\E(V_t^{3/2}) \over \E^{d+2}(V_t^{1/2})}\right],
\end{equation}
which, for all $t\geq 0$, must be equal to a constant $\beta(\alpha)$.
%Note  that if $W$ leads to $\{V_t: t\ge 0\}$ which satisfy the condition that (\ref{E4}) must be constant, then
%a random variable $W'$ with density $f_{W'}(w)$ proportional to
%$e^{-wq} f_W(w)$ for some $q>0$ leads to another family $\{V'_t\}$
%with density proportional $e^{-tv}v^{1/2}f_{W'}(v)$ and with $V'_0=V'$ that
%also has the property that (\ref{E4}) and (\ref{F3}) are constant.
%Without loss of generality we can assume that $V$ has finite positive integer moments.
%Otherwise the arguments below shift to $V\gets V'$ where $V'$ is based on
%some $q>0$.
This implies the following recursive relation for the moments of $V_t$: if we know that (\ref{E4}) is constant for $d=k$ and $d=4$, then it is also constant for $d=k+2$. Thus, it is sufficient to show that (\ref{E4}) is constant for $d = 3$ and $d = 4$, or $d = 3$ and $d = 5$.

Let us now consider case a), where the copula $C$ is a simplified PCC in $d\geq4$. From the three-dimensional marginal distributions, we obtain that (\ref{E4}) is constant for $d=3$. By conditioning on $X_4=x_4, X_3=x_3$, we conclude with a similar calculation as for (\ref{E2}) that
$$ \frac{E[V_t^2] \cdot E^2[V_t^{1/2}]}{E^{3}[V_t]} = const. $$
This, together with (\ref{E4}) being constant for $d=3$ implies that (\ref{E4}) is constant for $d=5$ and thus for all $d\geq 3$.

In case b), where the copula is a simplified PCC for all $\Sigma$, we know that for a three-dimensional marginal distribution, (\ref{F3}) holds for all $\alpha \geq 0$:
 $$-{\alpha\E[\{W^*(x_3)\}^{4/2}] \over \E^{4}[{W^*}^{1/2}(x_3)]}
   + {\E[\{W^*(x_3)\}^{2/2}]\E[\{W^*(x_3)\}^{3/2}] \over 
   \E^{5}[{W^*}^{1/2}(x_3)]}=\beta(\alpha).$$
Thus,  $\E[V_t^{2}]/\E^{4}[V_t^{1/2}]$ and $\E[V_t]\E[V_t^{3/2}] / \E^{5}[V_t^{1/2}]$ must be constants over $t$.
Together with (\ref{E4}) being constant in $t$ for $d=3$, this means that (\ref{E4}) is constant for $d=4$ and thus for all $d\geq 3$.
Note that, more precisely, we only require two different values of $\alpha$ in (\ref{F3}) for the argument above.
 
Summing up, we obtain that with $m=(d-1)/2$ the moments of $V_t$ are connected to the moments of $V$ via
$$\E[V_t^m]=a^m(t)\E[V^m]$$ for all $t>0$, $m=1,1.5,2,2,5,\ldots$,
where $a(t)=\E^2[V_t^{1/2}]/\E^2[V^{1/2}]$.

With all of the positive integer moments of $V$ and $V_t$ existing, the Laplace transforms of
$V$ and $V_t$, for $0\le s\le s_t$, where the constant $s_t$ may depend on $t$, can be written as
  $$\varphi_V(s)=1+\sum_{i=1}^\oo (-1)^i \E[V^i]s^i/i!,$$
  $$\varphi_{V_t}(s)=1+\sum_{i=1}^\oo (-1)^i \E[V^i_t]s^i/i!
  = 1+\sum_{i=1}^\oo (-1)^i \E[V^i]s^ia^i(t)/i!.  $$
Hence $\varphi_{V_t}(s)=\varphi_{V}(s a(t))$ in a neighborhood of 0
for the Taylor series expansion of the LTs about 0.
By \citeN[Section VII.6]{Feller}, the Taylor series in a positive neighborhood of 0 uniquely
determines the distribution. Hence $V_t=a(t)V$ for $t>0$.
The combination that Laplace transform tilting of the density leads to a scale-changed
random variable, implies that $V$ has a gamma density \cite[p. 576, Theorem 18.B.6]{MarshallOlkin}.
Hence, also $W$ is Gamma distributed, and the corresponding scale mixture is the multivariate t-distribution. \hfill $\square$

\section{Trivariate Marshall-Olkin (MO) copula}\label{section:mo}
To determine the bivariate conditional distributions of a three-dimensional MO copula, we work with the parameterization of \citeN[Chapter 3]{mai2012}. The three-dimensional MO copula is the survival copula of the rvs defined as
\begin{equation*}
\begin{split}
&X_1:=\min \Big\{ E_1, E_{12}, E_{13}, E_{123} \Big\}\\
&X_2:=\min \Big\{ E_2, E_{12}, E_{23}, E_{123} \Big\}\\
&X_3:=\min \Big\{ E_3, E_{13}, E_{23}, E_{123} \Big\},
\end{split}
\end{equation*}
where $E_I , I\subset \{1,2,3\}$ are independent and exponentially
distributed. For simplicity let us assume that all rate parameters are
equal, i.e., $\lambda_I = \lambda$, $\forall I\subset \{1,2,3\}$. This implies that $P(X_3=E_I\vert X_3=x_3)=1/4$ for all $E_I , I\subset \{1,2,3\}$, and thus we can determine the conditional survival distribution of $X_1,X_2\vert X_3=x_3$ as
\begin{equation*}
\begin{split}
\bar F(x_1,x_2\vert x_3)&=P(X_1 >x_1, X_2 > x_2 \vert X_3=x_3)\\
&=\sum_{3\in I \vert I\subset \{1,2,3\}} P(X_1 >x_1, X_2 > x_2 \vert X_3=x_3, X_3=E_I) \cdot P(X_3=E_I\vert X_3=x_3)\\
&=\frac{1}{4} \cdot \sum_{3\in I \vert I\subset \{1,2,3\}} P(X_1 >x_1, X_2 > x_2 \vert X_3=x_3, X_3=E_I).
\end{split}
\end{equation*}
If $X_3=x_3$ and $X_3=E_{123}$, we know that $x_1,x_2\geq E_{123} \geq E_{23},E_{13}$. Using this,
\begin{equation*}
\begin{split}
P(X_1 >x_1&, X_2 > x_2 \vert X_3=x_3, X_3=E_{123})\\
&  =\mathbbmss{1}_{x_1,x_2 \leq x_3} P\big(E_1 > x_1, E_2 >x_2, E_{12} > \max(x_1,x_2)\big), \\
\end{split}
\end{equation*}
since $X_2 > x_2$ implies $E_2,E_{12} > x_2$ and $X_1 > x_1$ implies $E_1,E_{12} > x_1$, which further yields $E_{12} > \max(x_1,x_2)$.

Let $X_{\{X>y\}}$ be the random variable with cdf $P(X\leq x \vert X>y)$, then we obtain
\begin{equation*}
\begin{split}
P(X_1 >x_1&, X_2 > x_2 \vert X_3=x_3, X_3=E_{23})\\
&  =\mathbbmss{1}_{x_2 \leq x_3} P\big(E_1 > x_1, E_2 >x_2, E_{12} > \max(x_1,x_2), {E_{123}}_{\{ E_{123} > x_3\}} > x_1,\\
 & \ \ \ \  {E_{13}}_{ \{ E_{13} > x_3\}} > x_1 \big), \\
P(X_1 >x_1 &, X_2 > x_2 \vert X_3=x_3, X_3=E_{3})\\
&  =P\big(E_1 > x_1, E_2 >x_2, E_{12} > \max(x_1,x_2), {E_{23}}_{\{ E_{23} > x_3\} } > x_2 ,\\
& \ \ \ \ {E_{13}}_{ \{ E_{13} > x_3\}}> x_1, {E_{123}}_{ \{ E_{123} > x_3 \} } > \max(x_1,x_2)\big), \\
P(X_1 >x_1 &, X_2 > x_2 \vert X_3=x_3, X_3=E_{13})\\
&  =\mathbbmss{1}_{x_1 \leq x_3} P\big(E_1 > x_1, E_2 >x_2, E_{12} > \max(x_1,x_2), {E_{123}}_{\{ E_{123} > x_3\} } > x_2  ,\\
 & \ \ \ \  {E_{23}}_{\{ E_{23} > x_3\}} > x_2 \big).
\end{split}
\end{equation*}
Putting together the conditional probabilities, the conditional survival function is
\begin{equation*}
\begin{split}
\bar F_{12\vert 3}(x_1,x_2\vert x_3)= &\frac{1}{4} e^{-\lambda(x_1+x_2+\max(x_1,x_2))}\\
\cdot \Big[ & \mathbbmss{1}_{x_1,x_2 \leq x_3} \cdot 4 + \mathbbmss{1}_{x_1,x_2 > x_3} e^{-\lambda(x_1+x_2+\max(x_1,x_2) - 3 x_3})\\ 
&+\mathbbmss{1}_{x_1 \leq x_3, x_2 > x_3}\cdot 2\cdot e^{-2\lambda (x_2-x_3)} 
+\mathbbmss{1}_{x_2 \leq x_3, x_1 > x_3}\cdot 2\cdot e^{-2\lambda (x_1-x_3)} \Big].
\end{split}
\end{equation*}
For the generalized inverse of the conditional survival function of e.g. $X_1\vert X_3=x_3$, this implies that
\begin{equation*}
\begin{split}
\bar F^{-1}(v \vert x_3)=\begin{cases}
-\ln(v) / (2\lambda) & v > e^{-2\lambda x_3}\\
x_3 &  e^{-2\lambda x_3}/2<v\leq e^{-2\lambda x_3}\\
[x_3 - \ln(2v)/(2\lambda)]/2 & v \leq e^{-2\lambda x_3}/2
\end{cases}
\end{split}
\end{equation*}
 Given this, we can evaluate the copula of the bivariate conditional distribution of $X_1,X_2\vert X_3=x_3$ to see that it depends on the value of $x_3$ in the areas where it is uniquely determined. To arrive at the conditional copula of the trivariate MO copula, and not of the corresponding distribution, we further have to transform $x_3$.
For the marginal survival function we obtain that $\bar F_3(x_3)=e^{-4\lambda x_3}$, with inverse $\bar F_3^{-1}(v)=-\log(v)/(4\lambda)$.

\section*{Acknowledgements}
Jakob St\"ober gratefully acknowledges financial support by TUM's Topmath program and a research stipend from Allianz Deutschland AG, while  Harry Joe is supported by an NSERC Discovery grant.
\bibliography{references}

\begin{thebibliography}{}

\bibitem[\protect\citeauthoryear{Aas, Czado, Frigessi, and Bakken}{Aas
  et~al.}{2009}]{aas2009}
Aas, K., C.~Czado, A.~Frigessi, and H.~Bakken (2009).
\newblock Pair-copula construction of multiple dependence.
\newblock {\em Insurance: Mathematics and Economics\/}~{\em 44}, 182--198.

\bibitem[\protect\citeauthoryear{Acar, Craiu, and Yao}{Acar
  et~al.}{2011}]{acar2011}
Acar, E.~F., R.~V. Craiu, and F.~Yao (2011).
\newblock Dependence calibration in conditional copulas: A nonparametric
  approach.
\newblock {\em Biometrics\/}~{\em 67\/}(2), 445--453.

\bibitem[\protect\citeauthoryear{{Ahmadi Javid}}{{Ahmadi
  Javid}}{2009}]{ahmadijavid2009}
{Ahmadi Javid}, A. (2009).
\newblock Copulas with truncation-invariance property.
\newblock {\em Comm. Statist. A. - Theor.\/}~{\em 38}, 3756--3771.

\bibitem[\protect\citeauthoryear{Almeida and Czado}{Almeida and
  Czado}{2011}]{almeida2011}
Almeida, C. and C.~Czado (2011).
\newblock Efficient {Bayesian} inference for stochastic time-varying copula
  models.
\newblock {\em Computational Statistics and Data Analysis. To appear\/}.

\bibitem[\protect\citeauthoryear{Ang and Chen}{Ang and Chen}{2002}]{ang2002b}
Ang, A. and J.~Chen (2002).
\newblock Asymmetric correlations of equity portfolios.
\newblock {\em Journal of Financial Economics\/}~{\em 63}, 443--494.

\bibitem[\protect\citeauthoryear{Bartram, Taylor, and Wang}{Bartram
  et~al.}{2007}]{bartram2007}
Bartram, S., S.~Taylor, and Y.~Wang (2007).
\newblock The euro and {E}uropean financial market dependence.
\newblock {\em Journal of Banking and Finance\/}~{\em 31}, 1461--1481.

\bibitem[\protect\citeauthoryear{Bauer, Czado, and Klein}{Bauer
  et~al.}{2011}]{bauer2011}
Bauer, A., C.~Czado, and T.~Klein (2011).
\newblock Pair-copula constructions for non {Gaussian} {DAG} models.
\newblock {\em To appear in the Canadian Journal of Statistics.\/}.

\bibitem[\protect\citeauthoryear{Bedford and Cooke}{Bedford and
  Cooke}{2001}]{bedford2001}
Bedford, T. and R.~Cooke (2001).
\newblock Probability density decomposition for conditionally dependent random
  variables modeled by vines.
\newblock {\em Ann. Math. Artif. Intell.\/}~{\em 32}, 245--268.

\bibitem[\protect\citeauthoryear{Bedford and Cooke}{Bedford and
  Cooke}{2002}]{bedford2002}
Bedford, T. and R.~Cooke (2002).
\newblock Vines --- a new graphical model for dependent random variables.
\newblock {\em Annals of Statistics\/}~{\em 30}, 1031--1068.

\bibitem[\protect\citeauthoryear{Cambanis, Huang, and Simons}{Cambanis
  et~al.}{1981}]{cambanis1981}
Cambanis, S., S.~Huang, and G.~Simons (1981).
\newblock On the theory of elliptically contoured distributions.
\newblock {\em Journal of Multivariate Analysis\/}~{\em 11}, 386--385.

\bibitem[\protect\citeauthoryear{Clayton}{Clayton}{1978}]{clayton1978}
Clayton, D.~G. (1978).
\newblock A model for association in bivariate life tables and its application
  in epidemiological studies of familial tendency in chronic disease incidence.
\newblock {\em Biometrica\/}~{\em 65\/}(1), 141--151.

\bibitem[\protect\citeauthoryear{Cook and Johnson}{Cook and
  Johnson}{1981}]{cook1981}
Cook, R.~D. and M.~E. Johnson (1981).
\newblock A family of distributions for modeling non-elliptically symmetric
  multivariate data.
\newblock {\em J. Roy. Statist. Soc. B\/}~{\em 43}, 210--218.

\bibitem[\protect\citeauthoryear{Di\ss{}mann, Brechmann, Czado, and
  Kurowicka}{Di\ss{}mann et~al.}{2011}]{dissmann2011}
Di\ss{}mann, J., E.~C. Brechmann, C.~Czado, and D.~Kurowicka (2011).
\newblock Selecting and estimating regular vine copulae and application to
  financial returns.
\newblock {\em submitted\/}.

\bibitem[\protect\citeauthoryear{Feller}{Feller}{1971}]{Feller}
Feller, W. (1971).
\newblock {\em An introduction to Probability Theory and Its Applications\/}
  (second ed.), Volume~2.
\newblock New York: Wiley.

\bibitem[\protect\citeauthoryear{Genest and MacKay}{Genest and
  MacKay}{1986}]{Genest1986}
Genest, C. and R.~MacKay (1986).
\newblock Copules archim\'{e}diennes et familles de lois bidimensionelles dont
  les marges sont donn\'{e}es.
\newblock {\em Canadian Journal of Statistics\/}~{\em 14}, 145--159.

\bibitem[\protect\citeauthoryear{Hob{\ae}k~Haff, Aas, and
  Frigessi}{Hob{\ae}k~Haff et~al.}{2010}]{haff2010b}
Hob{\ae}k~Haff, I., K.~Aas, and A.~Frigessi (2010).
\newblock On the simplified pair-copula construction --- simply useful or too
  simplistic?
\newblock {\em Journal of Multivariate Analysis\/}~{\em 101}, 1296--1310.

\bibitem[\protect\citeauthoryear{Hob{\ae}k~Haff}{Hob{\ae}k~Haff}{2012}]{haff2010}
Hob{\ae}k~Haff, I.~H. (2012).
\newblock Parameter estimation for pair-copula constructions.
\newblock {\em To appear in Bernoulli.\/}.

\bibitem[\protect\citeauthoryear{Joe}{Joe}{1996}]{joe1996}
Joe, H. (1996).
\newblock Families of $m$-variate distributions with given margins and
  $m(m-1)/2$ bivariate dependence parameters.
\newblock In {L. R{\"{u}}schendorf and B. Schweizer and M. D. Taylor} (Ed.),
  {\em Distributions with Fixed Marginals and Related Topics}, Volume~28,
  Hayward, CA, pp.\  120--141. Inst. Math. Statist.

\bibitem[\protect\citeauthoryear{Joe}{Joe}{1997}]{Joe}
Joe, H. (1997).
\newblock {\em Multivariate Models and Dependence Concepts}.
\newblock Chapman \& Hall, London.

\bibitem[\protect\citeauthoryear{Kimeldorf and Sampson}{Kimeldorf and
  Sampson}{1975}]{kimeldorf1975}
Kimeldorf, G. and A.~R. Sampson (1975).
\newblock Uniform representations of bivariate distributions.
\newblock {\em Comm. Statist.\/}~{\em 4}, 617--627.

\bibitem[\protect\citeauthoryear{Lindskog, McNeil, and Schmock}{Lindskog
  et~al.}{2003}]{lindskog2003}
Lindskog, F., A.~McNeil, and U.~Schmock (2003).
\newblock Kendall's tau for elliptical distributions.
\newblock {\em Credit Risk: Measurement, Evaluation and Management;
  Physica-Verlag, Heidelberg\/}, 149--156.

\bibitem[\protect\citeauthoryear{Longin and Solnik}{Longin and
  Solnik}{2001}]{longin2001}
Longin, F. and B.~Solnik (2001).
\newblock Extreme correlations in international equity markets.
\newblock {\em Journal of Finance\/}~{\em 56}, 649--676.

\bibitem[\protect\citeauthoryear{Mai and Scherer}{Mai and
  Scherer}{2012}]{mai2012}
Mai, J.-F. and M.~Scherer (2012).
\newblock {\em Simulating Copulas: Stochastic Models, Sampling Algorithms, and
  Applications}.
\newblock to appear at World Scientific.

\bibitem[\protect\citeauthoryear{Mardia}{Mardia}{1962}]{mardia1962}
Mardia, K.~V. (1962).
\newblock Multivariate {P}areto distributions.
\newblock {\em Ann. Math. Statist.\/}~{\em 33}, 1008--1015.

\bibitem[\protect\citeauthoryear{Marshall and Olkin}{Marshall and
  Olkin}{2007}]{MarshallOlkin}
Marshall, A.~W. and I.~Olkin (2007).
\newblock {\em Life Distributions: Structure of Nonparametric, Semiparametric
  and Parametric Families}.
\newblock New York: Springer.

\bibitem[\protect\citeauthoryear{McNeil and Ne\v{s}lehov\'a}{McNeil and
  Ne\v{s}lehov\'a}{2009}]{mcneil2009}
McNeil, A. and J.~Ne\v{s}lehov\'a (2009).
\newblock Multivariate {A}rchimedean copulas, $d$-monotone functions and
  $l_1$-norm symmetric distributions.
\newblock {\em The Annals of Statistics\/}~{\em 37}.

\bibitem[\protect\citeauthoryear{Mesfioui and Quessy}{Mesfioui and
  Quessy}{2008}]{mesfioui2008}
Mesfioui, M. and J.-F. Quessy (2008).
\newblock Dependence structure of conditional archimedean copulas.
\newblock {\em Journal of Multivariate Analysis\/}~{\em 99}, 372--385.

\bibitem[\protect\citeauthoryear{Min and Czado}{Min and Czado}{2010}]{min2010}
Min, A. and C.~Czado (2010, Fall).
\newblock Bayesian inference for multivariate copulas using pair-copula
  constructions.
\newblock {\em Journal of Financial Econometrics\/}~{\em 8\/}(4), 511--546.

\bibitem[\protect\citeauthoryear{M{\"u}ller and Scarsini}{M{\"u}ller and
  Scarsini}{2005}]{mueller2005}
M{\"u}ller, A. and M.~Scarsini (2005).
\newblock Archimedean copulae and positive dependence.
\newblock {\em Journal of Multivariate Analysis\/}~{\em 93}, 434--445.

\bibitem[\protect\citeauthoryear{Nelsen}{Nelsen}{2006}]{Nelsen}
Nelsen, R.~B. (2006).
\newblock {\em An Introduction to Copulas}.
\newblock Springer, New York.

\bibitem[\protect\citeauthoryear{Nikoloulopoulos, Joe, and Li}{Nikoloulopoulos
  et~al.}{2009}]{nikoloulopoulos2009}
Nikoloulopoulos, A., H.~Joe, and H.~Li (2009).
\newblock Extreme value properties of multivariate t copulas.
\newblock {\em Extremes\/}~{\em 12}, 129--148.
\newblock 10.1007/s10687-008-0072-4.

\bibitem[\protect\citeauthoryear{Patton}{Patton}{2006}]{patton2006}
Patton, A. (2006).
\newblock Modelling asymmetric exchange rate dependence.
\newblock {\em International Economic Review\/}~{\em 47\/}(2), 527--556.

\bibitem[\protect\citeauthoryear{Sklar}{Sklar}{1959}]{sklar}
Sklar, M. (1959).
\newblock Fonctions de r\'epartition \`a $n$ dimensions et leurs marges.
\newblock {\em Publ. Inst. Statist. Univ. Paris\/}~{\em 8}, 229--231.

\bibitem[\protect\citeauthoryear{St\"ober and Czado}{St\"ober and
  Czado}{2011}]{stoeber2011c}
St\"ober, J. and C.~Czado (2011).
\newblock Detecting regime switches in the dependence structure of high
  dimensional financial data.
\newblock {\em preprint\/}.

\bibitem[\protect\citeauthoryear{Takahasi}{Takahasi}{1965}]{takahasi1965}
Takahasi, K. (1965).
\newblock Note on the multivariate {B}urr's distribution.
\newblock {\em Ann. Inst. Statist. Math.\/}~{\em 17}, 257--260.

\end{thebibliography}
\end{document}